\documentclass[twoside]{article}

\usepackage[preprint]{aistats2026}
\usepackage{natbib}
\usepackage[utf8]{inputenc} 
\usepackage[T1]{fontenc}    
\usepackage{hyperref}       
\usepackage{url}            
\usepackage{booktabs}       
\usepackage{amsfonts}       
\usepackage{nicefrac}       
\usepackage{microtype}      
\usepackage{xcolor}         

\usepackage{microtype}
\usepackage{graphicx}
\usepackage{subfigure}
\usepackage{booktabs} 
\usepackage{thm-restate}
\usepackage{cancel}
\usepackage{wrapfig}

\usepackage{hyperref}

\usepackage{amsmath}
\usepackage{amsthm}
\usepackage{amssymb}
\usepackage{mathtools}

\usepackage{notation}

\usepackage[capitalize,noabbrev]{cleveref}

\theoremstyle{plain}
\newtheorem{theorem}{Theorem}[section]

\newtheorem{lemma}[theorem]{Lemma}

\theoremstyle{definition}
\newtheorem{definition}[theorem]{Definition}
\newtheorem{assumption}[theorem]{Assumption}
\theoremstyle{remark}

\RequirePackage{algorithm}
\RequirePackage{algorithmic}
\allowdisplaybreaks

\begin{document}

%

%

\twocolumn[

\aistatstitle{Online Decision-Making in Tree-Like Multi-Agent Games with Transfers}


\aistatsauthor{ Antoine Scheid \And Etienne Boursier \And  Alain Durmus \And Eric Moulines \And Michael I. Jordan}

\aistatsaddress{ CMAP, CNRS \\ École polytechnique
\And 
Inria Saclay, LMO \\
Universit\'e Paris Saclay
\And
CMAP, CNRS \\
École polytechnique
\And
CMAP, CNRS \\
École polytechnique
\And
Ecole Normale Sup\'erieure \\
Inria Paris \\
University of California, Berkeley} ]

\begin{abstract}
The widespread deployment of Machine Learning systems raises challenges, such as dealing with interactions or competition between multiple learners. In this work, we investigate multi-agent sequential decision-making in the context of principal–agent interactions structured as a tree. In this problem, the reward of a player is influenced by the actions of her children, who are all self-interested and non-cooperative, hence the complexity of making good decisions. Our main finding is that it is possible to steer all the players towards the globally optimal set of actions by simply allowing single-step transfers between them. A \textit{transfer} is established between a principal and one of her agents: the principal actually offers the proposed payment if the agent picks the recommended action. The analysis poses specific challenges due to the intricate interactions between the nodes of the tree and the propagation of the regret within this tree. To the best of our knowledge, it is the first formulation of multilayered principal–agent interactions via by incentives. Considering a bandit setup, we propose algorithmic solutions for the players to end up being no-regret with respect to the optimal pair of actions and incentives. In the long run, allowing transfers between players makes them act as if they were collaborating together, although they remain self-interested and non-cooperative: \textit{transfers restore efficiency}.
\end{abstract}

\section{Introduction}

The increasing deployment of Machine Learning (ML) systems in real-world environments introduces specific challenges, particularly regarding interactions and competition among self-interested learning agents. This has led to a growing literature that incorporates economic considerations into ML settings. A large part of these works focuses on \textit{mechanism design} \citep{laffont2009theory, nisan1999algorithmic, myerson1989mechanism}, where a principal (\textit{she}) wants an agent (\textit{he}) to perform a specific task and needs to incentivize the agent so he participates. Usually, the mechanism is a payment based on the agent's behavior. The majority of this literature assumes that one principal interacts with one agent or more. Unfortunately, real-world scenarios generally present much more intricate structures. When an agent accepts a payment from a principal, it is likely that this one agent then delegates a part of his effort to another agent, and so on and so forth. Extending this reasoning reveals a natural \textit{hierarchy of players}, forming a chain of interactions through principal-agent relations. It raises an issue about the behavior of players organized in a general structure of relations, who interact together through transfers in a repeated game; put differently:

\textit{Can self-interested learning agents converge to social welfare maximization through online incentives with their neighbors?}

\looseness=-1
Although the analysis of equilibrium in learning games has been extensively studied \citep{wang2002reinforcement, celli2020no, chakrabarti2024efficient, daskalakis2024efficient}, these settings solely involve one principal, her agents, and sometimes a central planner. As a first step towards solving our question in a general decentralized structure with no mediator, we assume a principal-agent structure of the players, organized as a \textit{tree} \citep[such structures are often studied in games, see,][]{zhang2022team, fiegel2023adapting}. The players sequentially take actions on a bandit instance, with their action influencing both their own reward and the reward of their principal. Our objective is to analyze the convergence of the entire set of players towards the global optimum, defined as the set of actions that maximizes the sum of all utilities. However, such convergence cannot occur endogenously, as individual players have no incentive to account for their principal's reward in their decision-making. To resolve this, we introduce \textit{transfers} -- also mentioned as \textit{incentives} -- (learned over time) from the principal to the agent, that help to align the players' utilities. This question is related to \textit{Contract Theory}, a field devoted to study how a principal can enforce an agent to take specific actions in order to obtain given outcomes. Contracting often presents two primary challenges:
\textit{Hidden Type}, also known as \textit{Adverse Selection}, which refers to the scenario where the principal cannot observe the agent's type; and \textit{Hidden Action} \citep[a.k.a. \textit{Moral Hazard}, see,][]{mirrlees1999theory}, which refers to the scenario where the principal cannot observe the actions taken by the agent. In this work, we are interested in scenari where the principal observes the agent's action.

In the field of Machine Learning, consider the example of delegated data collection. In such case, agents must be incentivized to participate \citep[see, e.g.,][]{zhan2020learning, zhan2021survey, karimireddy2022mechanisms, liu2022contract, capitaine2024unravelling}. If a principal pays agents to perform a learning/data collection task, we would model it as selecting an action. These agents, in turn, may delegate part of their task to data servers, offering them contracts to collect and process data—which again corresponds to choosing an action in a bandit instance. This cascading delegation naturally leads to a hierarchical structure of decision-making and incentives. This example highlights the need to develop models that handle complex and multi-layered incentivization mechanisms. While prior works studied the delegation of learning tasks between a principal and an agent through the lens of contract theory \citep[see, e.g.,][]{ananthakrishnan2024delegating, saig2024delegated}, this is -- to our knowledge -- the first work to consider a layered incentivization structure between principals and agents.

Therefore, we present a general formulation of an extensive form game between players organized as a tree with directed relations. At each step, a node (a player) observes the \textit{incentive} (a recommended action and an associated payment) offered by her parent. Then, the player takes an action on her multi-armed bandit instance and offers an incentive to her children. The reward of a node depends on her action as well as the actions taken by her children. Note that with regards to this delegated mechanism, the reward of a player $\vl$ is not \textit{directly} influenced by actions of players further down the tree. However, the actions of players several layers below affect their parents’ actions, and this influence propagates upward, ultimately leading to an \textit{indirect} influence on $\vl$.

\looseness=-1
\textbf{Contributions and Techniques.} After introducing our setting, we present $\alg$: \texttt{Multi-Agent Incentivized Learning}, an algorithm that enables the players to learn the optimal incentives to offer before executing a bandit subroutine on the updated problem instance, which accounts for the estimated payments. Each player must wait until their children have sufficiently learned and are best-responding with enough accuracy before initiating the process of exploring their preferences. An important challenge is to ensure that children almost best-response with high probability, since learning is not possible otherwise for the parent. Finally, we derive regret bounds for the players as well as an upper bound for the \textit{social welfare regret} that evolves as $o(T)$. From a technical perspective, the regret bound results are achieved by decomposing the regret of each player into three components: the error due to her own \textit{suboptimal actions}, the error arising from \textit{inaccuracies in the incentives}, and the error caused by the \textit{suboptimal decisions taken by her agents}. Interestingly, by controlling individual regrets while accounting for transfers, we are able to manage the collective regret of all players, effectively treating them as if they were collaborating. \textit{It reveals that it is possible to orchestrate selfish players without the need for a central server, solely by enabling transfers.}


\section{Related Works}

The study of social welfare (also referred to as global utility) in interactions between a principal and an agent \citep{grossman1992analysis, bendor2001theories} has been a long-standing area of research in the mechanism design literature that recently benefited from data-driven approaches \citep{bergemann2024data} and, more recently, applications of large language models \citep{duetting2024mechanism, dubey2024auctions}. The study of learning through sequential interactions has also gained attention \citep{han2024learning, kolumbus2024paying, zuo2024new, zuo2024optimizing}, especially with the recent rise of algorithmic contract design \citep{dutting2022combinatorial, dutting2024algorithmic, bacchiocchi2023learning, guruganesh2024contracting, bollini2024contracting, collina2024repeated}. While prior work has examined interactions between a principal and multiple agents \citep{dutting2023multi}, to the best of our knowledge, \textit{this paper is the first} to address the more complex scenario of a nested chain of interacting principals and agents. From a statistical point of view, our work relies on ideas from the bandit literature \citep{lattimore2020bandit, slivkins2019introduction} but diverges significantly from the literature on multi-player or competing bandits \citep{liu2020competing, boursier2022survey}, with simultaneous actions and no transfers.


\begin{figure}[t]
  \centering
  \includegraphics[width=\columnwidth]{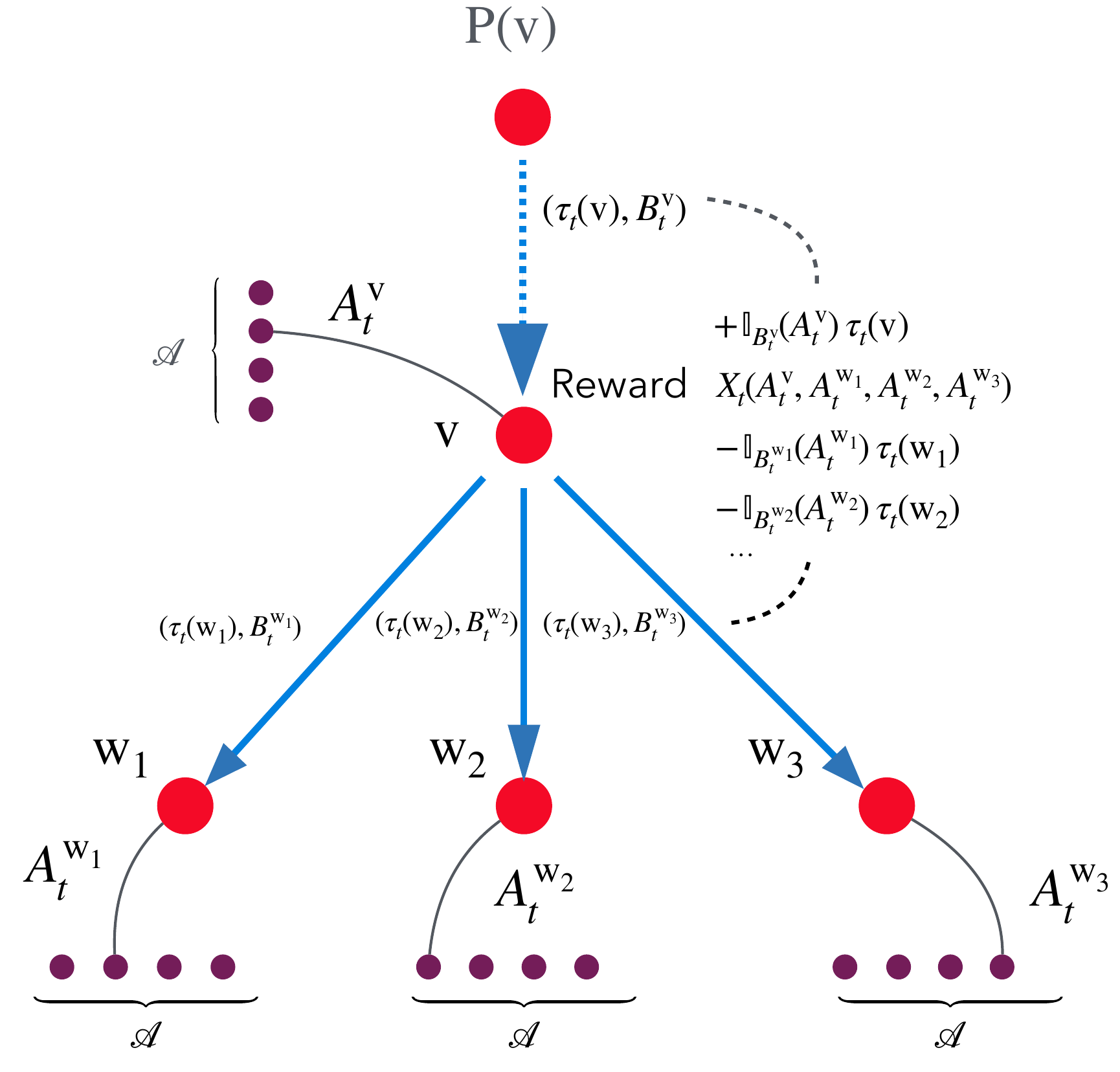}
  \caption{Illustration of the game structure: ...}
  \label{Illustration of the game structure: each node of the tree (in red) represents a player and the blue edges correspond to the principal/agent interactions between the players.}
\end{figure}

More importantly, our work also relies on the interactions between principal-agent bandit learners, as they are explored in \citet{scheid2024learning, scheid2024incentivized}, but their model only considers one principal and one agent. The concept of interacting bandit players has been first introduced in \citet{dogan2023estimating,dogan2023repeated}, with notable extensions to Markov Decision Processes \citep{ivanov2024principal, wu2024contractual} and frameworks where agents are constrained in probability to stay nearby their best responses \citep{liu2024principal}. The extension of our scenario to MDPs is definitely a long step ahead since \citet{ben2024principal} show that even in the two-player case—where a principal sends a bonus reward to an agent whose actions influence the principal’s utility—computing the principal’s optimal policy is NP-hard.

The study of principal–agent problems through the lens of learning is a growing field. Among others, \citet{wu2025learning} consider a principal interacting with strategic agents who have private information, while \citet{liu2024principal} examine a related setting in which agents take random exploratory steps. \citet{widmer2025steering} extend this line of work to mean-field games. Our work advances this literature a step further by introducing layered interactions among players who simultaneously act as both principals and agents.

\section{Setting}

\textbf{Model.} We consider a set of players $\trV$ organized as a tree $\cT = (\trV, \trE)$ where $\trE$ represents the set of edges between the players. These players engage in an online learning game over a horizon of $T$ steps. At each time step $t \in \iint{1}{T}$, each player $\vl \in \trV$ selects an action $A_t^\vl$ from a common action set $\cA$ of size $K$ (it is equivalent to consider the same action set or distinct action sets for the players).

\looseness=-1
We assume that each player is an expected-utility maximizer and does not strategically manipulate the others. Although it is possible for a learning agent to achieve its Stackelberg value through strategic interactions \citep{goktas2022robust, zhao2023online, collina2024efficient, haghtalab2024calibrated}, this is not always the case \citep[as shown in][]{ananthakrishnan2024knowledge}. Furthermore, we emphasize that our focus here is not on achieving a \textit{Stackelberg optimum} since we consider a game with transfers.

\looseness=-1
A key distinction from existing literature is that the players in our framework are organized in a hierarchical tree structure with depth $\mD$ and breadth $\mB$, reflecting \textit{principal-agent} relationships. Each player $\vl$ (represented as a node in the tree) serves as the principal for her children, denoted $\ch(\vl)$, while simultaneously acting as an agent for her parent, $\pa(\vl)$. Without loss of generality, we assume that each player has exactly $\mB$ children. The tree is indexed by its depth $d \in \iint{1}{\mD}$, where $d = 1$ corresponds to the leaf nodes, and $d = \mD$ represents the root. The total number of vertices in the tree is $\Card(\trV) = \sum_{k=1}^\mD \mB^k = (\mB^{\mD+1} - 1)/(\mB - 1)$. For the sake of clarity, we refer to the specific player under consideration as $\vl$, her children as $\wl$, and her parent as $\ul$ throughout the analysis.

\textbf{Reward.} Within this tree structure, the reward of a player $\vl$ at time $t$ depends on her action $A_t^\vl$ and the actions of all her children, denoted as $A_t^{\ch(\vl)} = (A_t^\wl)_{\wl \in \ch(\vl)}$. The reward of player $\vl$ at round $t$ is given by $X_t^\vl(A_t^\vl, A_t^{\ch(\vl)})$, which is drawn at random following:
\begin{equation}\label{equation:definition_bandit_reward}
X_t^\vl(a, a^{\ch(\vl)}) = \theta^\vl(a, a^{\ch(\vl)}) + z_t^\vl \eqsp,
\end{equation}
where $(z_t^\vl)_{t \in \N^\star}$ are i.i.d. zero-mean sub-Gaussian random variables, and $(\theta^\vl(a, a^{\ch(\vl)}))_{(a, a^{\ch(\vl)}) \in \cA^{\mB+1}} \in \R^{K^{\mB+1}}$. For any $(a, a^{\ch(\vl)}) \in \cA^{\mB+1}$, the expected reward is given by: $\theta^\vl(a, a^{\ch(\vl)}) = \E[X_t^\vl(a, a^{\ch(\vl)})]$, which represents the mean reward of player $\vl$ when the actions $(a, a^{\ch(\vl)})$ are played. The quantities $(\theta^\vl(a, a^{\ch(\vl)}))_{(a, a^{\ch(\vl)}) \in \cA^{\mB+1}}$ are unknown to the players, whose objective is to learn the actions with the highest mean rewards. This is achieved through a sequential selection of actions. We impose the following assumption on the mean rewards, which is equivalent to assuming bounded rewards, as they can always be rescaled.

\begin{assumption}\label{assumption:bounded_rewards}
    For any player $\vl \in \trV$ and any actions $a \in \cA, a^{\ch(\vl)} \in \cA^\mB$, we have that $\theta^\pl(a, a^{\ch(\vl)}) \in [0,1]$.
\end{assumption}

\looseness=-1
\textbf{On the Necessity of Transfers.} In this setup, as in any game, the primary difficulty lies in the fact that a player's reward depends on the actions of other players. Without additional mechanisms, the global utility of the participants is unlikely to be maximized: some players (e.g., leaf nodes) may achieve reasonably good utility, while others may gain little to no utility \citep[see e.g.,][for two players]{scheid2024learning}. To address this issue, we allow players to establish payments from one to another.

We define an \textit{incentive} = \textit{short-term contract} as a pair $(B_t^\wl, \icv_t(\wl)) \in \cA \times \R_+$, where player $\vl$ offers a transfer $\icv_t(\wl)$ to her child $\wl$ if the latter plays the recommended action $B_t^\wl$. Player $\wl$ can either accept the incentive and play $B_t^\wl$ to receive the transfer $\icv_t(\wl)$ or reject it and play any action of his choice, forfeiting the transfer. At each step $t \in \iint{1}{T}$, player $\vl$ first observes the incentive $(B_t^\vl, \icv_t(\vl))$ offered by her parent $\pa(\vl)$. Then, $\vl$ chooses the action of her choice $A_t^\vl$ and offers an incentive $\cC_t^\wl = (B_t^\wl, \icv_t(\wl))$ to any child $\wl \in \ch(\vl)$, which includes recommended actions $B_t^{\ch(\vl)} = (B_t^\wl)_{\wl \in \ch(\vl)}$ and transfers $(\icv_t(\wl))_{\wl \in \ch(\vl)}$. The resulting utility is:
\begin{equation}\label{equation:utility}
\begin{aligned}
    \ut_t^\vl 
    & = X_t^\vl(A_t^\vl, A_t^{\ch(\vl)}) + \1_{B_t^\vl}(A_t^\vl) \cdot \icv_t(\vl) \\
    & \quad - \sum_{\wl \in \ch(\vl)} \1_{B_t^\wl}(A_t^\wl) \cdot \icv_t(\wl) \eqsp,
\end{aligned}
\end{equation}
which accounts for the incentive received from $\pa(\vl)$. At the end of each round, in addition to observing the bandit feedback $X_t^\vl(A_t^\vl, A_t^{\ch(\vl)})$, player $\vl$ also observes her children's actions $(A_t^\wl)_{\wl \in \ch(\vl)}$ but nothing more.

\textbf{Extensive Form Game.} We can see that at each round, the players play an extensive form game, where the root player plays first, and actions are then played in a sequential manner while descending the tree until the leaves who observe their offered incentive and then pull an arm.

\textbf{Challenges.} Note that $\vl$ also acts as an agent for her parent $\pa(\vl)$ and receives an incentive $(B_t^{\vl}, \icv_t(\vl))$. As a result, $\vl$ faces a three-fold trade-off: balancing the contract  received from her parent, the action she chooses, and the contracts she offers to her children. 
\section{Algorithm and Results}

\textbf{Players' Policies.} Consider a player $\vl$, whose \textit{history} is defined as $\cH_0^\vl = \varnothing$, and for any $t > 0$:
\begin{equation*}\label{equation:definition_history}     \cH_t^\vl = \cH_{t-1}^\vl \cup \{U_t^\vl, B_t^\vl, \icv_t(\vl), A_t^{\ch(\vl)}, X_t(A_t^\vl, A_t^{\ch(\vl)})\},
\end{equation*}
where $((U_t^\vl)_{t \in \N^\star})_{\vl \in \trV}$ are families of independent uniform random variables in $[0,1]$ used for policy randomization. A \textit{policy} for player $\vl \in \trV$ is defined as a mapping:
\begin{equation}\label{equation:definition_policy}
\begin{aligned}
    \pi^\vl &\colon (U_{t+1}^\vl, \cH_t^\vl, B_{t+1}^\vl, \icv_{t+1}(\vl))  \\
    & \quad \mapsto (B_{t+1}^{\ch(\vl)}, \icv_{t+1}(\ch(\vl)), A_{t+1}^\vl) \eqsp.
\end{aligned}
\end{equation}

In this setup, player $\vl$ must determine her optimal action $A^\vl$ while simultaneously ensuring that her children $\ch(\vl)$ perform actions $B^{\ch(\vl)}$ that maximize her mean reward $\theta^\vl(A^\vl, B^{\ch(\vl)})$. To achieve this, the contract $(B^{\ch(\vl)}, \icv(\ch(\vl)))$ must incentivize the children to follow the recommended actions. The transfer $\icv(\ch(\vl))$ poses a \textit{critical trade-off}: if set too high, it reduces $\vl$’s utility; if set too low, the children may ignore the recommendations. Designing such a policy requires sophisticated algorithmic strategies to balance these competing objectives effectively.

\textbf{Optimal Incentives.} For any players $\vl \in \trV$ and $\wl \in \ch(\vl)$, any action $b \in \cA$ that $\vl$ wants to be played by $\wl$, we define the optimal incentive $(b,\icvs_b(\wl))$ as the minimal transfer in hindsight that $\vl$ needs to offer to $\wl$ so that $b$ is in expectation the best action that $\wl$ could play.

The definition of $\icvs$ appears by induction starting with the optimal incentives for the leaves and then ascending up to the root. Consider a player $\vl$ at depth $2$ who interacts with one of her children $\wl$ at depth $1$ (a leaf) and assume that $\wl$ is best-responding. Since $\wl$ is a leaf and has no agent, $\ch(\wl) = \varnothing$. Hence, if $\vl$ offers an incentive $\cC^\wl=(B^\wl,\icv(\wl))$ to $\wl$ and wants to make $B^\wl$ better than any other action, we need to have for any $a \ne B^\wl$, we must have $\theta^\ql(B^\wl) + \icv(\wl) > \theta^\ql(a)$,
which gives, since we want $\icvs_{B^\wl}(\wl)$ as small as possible
\begin{equation}\label{equation:transfer_last_layer}
    \icvs_{B^\wl}(\wl) = \max_{a \in \cA} \theta^\wl(a) - \theta^\wl(B^\wl) \eqsp.
\end{equation}
Therefore, for a set of actions $(A_t^\vl,B^{\ch(\vl)}_t)$: if $\vl$ knows the optimal incentives $(b,\icvs_b(\wl))_{b \in \cA, \wl \in \ch(\vl)}$ and her agents are best-responding, she obtains an expected utility at step $t$
\begin{equation}\label{equation:utility_last_layer}
\begin{aligned}
    \E[\cU_t^\vl
    ] = \theta^\vl(A_t^\vl,B^{\ch(\vl)}_t) + \1_{B^\vl_t}(A^\vl_t)\icv_t(\vl) - \sum_{\wl \in \ch(\vl)} \icvs_{B^\wl_t}(\wl) \eqsp.
\end{aligned}
\end{equation}

\looseness=-1
\textbf{Optimization Problem in Hindsight.} The quantities in \eqref{equation:transfer_last_layer} and \eqref{equation:utility_last_layer} are straightforward to compute for players at the leaves of the tree, but the task becomes significantly more complex for nodes in the core. Consider a player $\vl \in \trV$ at depth $d \geq 3$. We assume fully rational players with perfect knowledge of the game, i.e.,  all children $\wl \in \ch(\vl)$ are aware of the optimal average utility $\reww(b)$ they can extract from any action $b \in \cA$. This utility is well-defined as long as the optimal incentive schemes are specified, as in \eqref{equation:transfer_last_layer}.

For a single round, player $\vl$'s objective is to identify a set of actions $(a^\star, b^{\star, \ch(\vl)})$ and propose an optimal incentive $\cC^\star(b^{\star, \ch(\vl)}) = (b^{\star, \ch(\vl)}, \icvs_{b^{\star, \ch(\vl)}}(\ch(\vl)))$ that solves the following optimization problem:
\begin{align}\label{equation:optimization_problem}
    &\max_{a, \tau} \; \theta^\vl(a, b^{\ch(\vl)}) - \sum_{\wl \in \ch(\vl)} \icv(\wl) \\
    & \nonumber \text{s.t.  } \icv(\wl) \in \R_+, \; b^\wl \in \argmax_{b \in \cA} \{\reww(b) + \1_{b^\wl}(b)\icv(\wl)\} \eqsp,
\end{align}
for all $\wl \in \ch(\vl)$. This problem requires knowledge of the optimal transfers $(b^\wl, \icvs_{b^\wl}(\wl))$ for each child $\wl \in \ch(\vl)$ to be solved exactly. The transfer $\icvs_{b^\wl}(\wl)$ is defined as the minimum utility that must be transferred from $\vl$ to $\wl$ to ensure that $\wl$ selects $b^\wl$ as their optimal action in hindsight.

\textbf{Social Welfare} is defined as the total reward accumulated across all players. \Cref{equation:optimization_problem} shows that the introduction of contracts incentivizes players to select actions that not only optimize their own rewards, but also contribute to the overall benefit of others. Actions that lead to a higher collective reward thus become less costly, \textit{aligning individual incentives with social welfare}.

\begin{restatable}{lemma}{optimalincentives}\label{lemma:optimal_incentives}
    For any player $\vl \in \trV$, the optimal incentives $\cC^\star(b^{\ch(\vl)}) = (b^{\wl}, \icvs_{b^{\wl}}(\wl))_{\wl \in \ch(\vl)}$ for any $b^{\ch(\vl)} = (b^\wl)_{\wl \in \ch(\vl)} \in \cA^\mB$ are given by $\icvs_{b^{\wl}}(\wl) = \max_{b \in \cA} \reww(b) - \reww(b^\wl),
    $ and the expected utility for the couple of actions $(a,b^{\ch(\vl)})$ with the optimal incentives $\cC^\star(b^{\ch(\vl)})$ are well-defined and unique, following \begin{equation}\label{equation:mu_utility}
    \mu^\vl(a, b^{\ch(\vl)}) = \theta^\vl(a, \b^{\ch(\vl)}) - \sum_{\wl \in \ch(\vl)} \icvs_{\b^\wl}(w) \eqsp,
    \end{equation}
    as well as the best utility $\vl$ can obtain from action $a$ and any couple $b^{\ch(\vl)} \in \cA^\mB$ as
    \begin{equation}\label{equation:mu_utility_star}
    \rewv(a) = \max_{b^{\ch(\vl)} \in \cA^\mB} \mu^\vl(a, b^{\ch(\vl)}) \eqsp.
\end{equation}
We also have that   $\sum_{\vl \in \trV} \max_{a^\vl \in \cA}
        \rewv(a^\vl) =  \max_{(a^\vl) \in \cA^{\trV}} \sum_{\vl \in \trV} 
        \theta(a^\vl,a^{\ch(\vl)})$.
\end{restatable}

The quantities arising in \Cref{lemma:optimal_incentives} are well-defined for the leaves for which it is easy to quantify (see \eqref{equation:transfer_last_layer}). Then, an induction allows to define these quantities for the upper layers (see the proofs in \Cref{appendix:omitted_proofs}). In the last equation from \Cref{lemma:optimal_incentives}, the left hand side represents the target of the players playing on their own while the right hand side represents the target of collaborative players. This lemma shows that if the \textit{behavior} of their children is known, computing explicitly the optimal policy for a node can be done. For the whole tree, the complexity of the computation is $\Card(V) \, K^{B+1}$. As stated formally in \Cref{lemma:spne_outcome} in the appendix, the action profiles (including the recommended actions and incentives) described in \Cref{lemma:optimal_incentives} correspond to the actions played by players whose strategy profile forms a \textit{subgame perfect Nash equilibrium} -- see \Cref{equation:definition_spne} in the appendix.

\textbf{Regret Definition.} Since we define $\mu^\vl(a,b^{\ch(\vl)})$ as the utility for player $\vl$ with action $a$ and her children playing actions $b^{\ch(\vl)}$, we can now define her (pseudo)-regret \citep[see,][]{auer2016algorithm, lattimore2020bandit} for the time window $\iint{s}{t}$, implying the optimal incentives in hindsight as well as the received contract, following
\begin{align}
    \nonumber \regret^\vl(s,t) & = \sum_{l=s}^{t} \max_{a, b^{\ch(\vl)}} \left\{\mu^\vl(a, b^{\ch(\vl)}) + \1_{B_l^\vl}(a) \, \icv_l(\vl) \right\} \\
    & \label{equation:definition_regret} \quad - \sum_{l=s}^t \1_{B_l^\vl}(A_l^\vl)\, \icv_l(\vl) \\
    & \nonumber - \sum_{l=s}^t \left\{ \theta^\pl(A_l^\vl, A_l^{\ch(\vl)}) - \sum_{\wl \in \up(\pl)} \1_{B_l^\wl}(A^\wl_l)\icv_l(\wl)\right\} .
\end{align}
This notion of regret for our game takes into account the contract $(B_t^\vl,\icv_t(\vl))$ that player $\vl$ receives from $\pa(\vl)$ at each step $l\in \iint{s}{t}$. This is coherent in the sense that if a huge transfer is offered to $\vl$ for some specific action $B^\vl$, it would be a bad move not to play it -- hence the presence of $\1_{B_l^\vl}(A_l^\vl)\, \icv_l(\vl)$ in the regret. Finally, we define $\regret^\vl$ on any window $\iint{s}{t} \subseteq \N$ because analyzing regret on specific time intervals is crucial in this work. For the ease of notation, we also define the regret over the whole horizon as
\begin{equation*}
    \regret^\vl(T) = \regret^\vl(1,T) \eqsp,
\end{equation*}
and we now introduce a key lemma that exposes the various sources of regret experienced by a player.

\begin{restatable}{lemma}{decompositionregret}\label{lemma:decomposition_regret}
  For any $s,t\in\N^*$ and $\vl \in \trV$, we can decompose the regret into several terms, following
    \begin{align*}
    &\regret^\vl(s,t) \leq \actionregret^\vl(s,t) + \paymentregret^\vl(s,t) + \deviationregret^\vl(s,t) \eqsp,
    \end{align*}
    where the \textit{action regret} is defined as
\begin{equation*}
\begin{aligned}
    \actionregret^\vl(s,t) & = \sum_{l=s}^{t} \max_{a \in \cA} \{\rewv_a + \1_{B_l^\vl}(a)\, \icv_l(\vl)\} \\
    & \quad - \sum_{l=s}^{t} \left\{\mu^\vl(A_l^\vl, B_l^{\ch(\vl)}) + \1_{B_l^\vl}(A_l^\vl)\, \icv_l(\vl))\right\} \eqsp,
    \end{aligned}
\end{equation*}
the \textit{payment regret} as: $$\paymentregret^\vl(s,t) = \sum_{l=s}^t \sum_{\wl \in \ch(\vl)} (\icv_{B_l^\wl}(\wl)-\icvs_{B_l^\wl}(\wl))_+$$
and the \textit{deviation regret} as: $$\deviationregret^\vl(s,t)=\sum_{l=s}^t (\theta^\vl(A_t^\vl, B_t^{\ch(\vl)}) -\theta^\vl(A_t^\vl, A_t^{\ch(\vl)}))_+ .$$
\end{restatable}

Note that the \textit{action regret} $\actionregret^\vl$ corresponds to player $\vl$'s error in the action she chooses to play. The \textit{payment regret} $\paymentregret^\vl$ corresponds to the payment excess of $\vl$ to her children with respect to the optimal transfers, while the \textit{deviation regret} $\deviationregret^\vl$ corresponds to how much the children deviated from the proposed contracts.

\textbf{Good Behavior.} The following assumption about the players' \textit{action regret} is the most decisive tool of this work. Intrinsically, our results hold because we are able to show that if the children of a node satisfy this assumption, then this node satisfies the assumption as well, with potentially different parameters. Similar assumptions are often done in the literature \citep[see, e.g.,][]{donahue2024impact} when online principal-agent interactions are considered.

\begin{restatable}{assumption}{boundedregretbis}\label{assumption:bounded_regret} Considering a node $\vl \in \trV$, there exists parameters $(\wait_\vl, \acst_\vl, \kappa_\vl, \zeta_\vl)$ such that for any steps $s > \wait_\vl$, $t \in \N^\star$ with $s+t\leq T$ and transfers from $\pa(\vl)$ such that $(B_l^\vl, \icv_l(\vl)) = (B_l^\vl, \icv_{B_l^\vl})$ for any step $l \in \iint{s+1}{s+t}$, then
\begin{align*}
    \actionregret^\vl(s+1,s+t) \leq \acst_\vl \, t^{\kappa_\vl} \eqsp,
\end{align*}
    with probability at least $1-1/t^{\zeta_\vl}$. We call $\wait_\vl$ the \textit{$\waittext$} associated with player $\vl$.
\end{restatable}
It is important to note that \Cref{assumption:bounded_regret} is solely a \textit{no-regret assumption} on the players. We write \Cref{assumption:bounded_regret} for the \textit{Action Regret} in order to avoid a dramatic \textit{propagation} of the errors from the leaves to the roots, as explained in the penultimate paragraph of the section. The literature generally considers single principal-agent interactions \citep{haghtalab2022learning,mansour2022strategizing}. This work is one of the first in the learning community to study nested interactions between players, as well as how the corresponding parameters evolve along the tree. This is why it requires some new techniques, such as \Cref{assumption:bounded_regret}. Consider a leaf $\vl \in \trV$: there are no children, hence $\ch(\vl) =\varnothing$ and $\rewv(a) = \theta^\vl(a)$ for any $a \in \cA$. Thus, her \textit{action regret} can be written
\begin{align*}
    \actionregret^\vl(s+1,s+t) & = \sum_{l=s+1}^{s+t} \max_{a \in \cA} \{\theta(a) + \1_{B_l^\vl}(a)\, \icv_l(\vl)\} \\
    & \quad - \sum_{l=s+1}^{s+t} \left\{\theta^\vl(A_l^\vl) + \1_{B_l^\vl}(A_l^\vl)\, \icv_l(\vl))\right\} ,
\end{align*}
and \citet[Proposition 2]{scheid2024learning} show that in that case, if $\vl$ uses a \texttt{UCB} subroutine (see \Cref{algorithm:ucb}), \Cref{assumption:bounded_regret} is satisfied with constants $(0,8\sqrt{K\log(KT^3)}, 1/2,2)$. It also holds with other generic bandit algorithms, such as \texttt{Explore Then Commit} for instance.

Our goal is to design an algorithm that is no-regret for a player $\vl$, when all of her children $\wl$ satisfy \cref{assumption:bounded_regret}. Additionally, we want the player $\vl$ to also satisfy \cref{assumption:bounded_regret} when following this algorithm. In that purpose, we introduce the algorithm $\alg$, with its pseudo-code in \cref{algorithm:algorithm}.

\textbf{Our Algorithm.} Learning from not well-behaving agents is impossible. This is why the first step of $\alg$ is to wait $\max_{\wl \in \ch(\vl)} \wait_{\wl}$ steps so the children best-respond almost all of the time to the contracts \citep[for instance in][the agent also approximately best-responds, which gives power to the principal]{lin2024persuading}. Since the children all satisfy \Cref{assumption:bounded_regret}, we are then ensured that they have a small regret. Thus, we use an epoched binary search-like procedure to get an upper estimate $(\icvest_b(\wl))_{b\in \cA,\wl \in \ch(\vl)}$ of the optimal payments, following \eqref{equation:estimated_incentives_definition}. This is achieved using the auxiliary procedure $\expsub$ (see \Cref{algorithm:search_subroutine}). Then, $\alg$ uses a black-box bandit subroutine (typically \texttt{UCB}, see, e.g., \Cref{algorithm:ucb}) which runs on a $K^{\mB+1}$-armed bandit instance with shifted rewards defined as $\theta^\vl(A,B^{\ch(\vl)})-\sum_{\wl \in \ch(\vl)} \icvest_{B^\wl}(\wl)$ and takes into account $\vl$'s parent incentives, also shifting the rewards. We call a \textit{black-box} subroutine any bandit algorithm that runs on a multi-armed bandit instance and is \textit{no-regret}. For a set of actions $(A^\vl,B^{\ch(\vl)})$ recommended by the black-box subroutine, $\alg$ offers an incentive $(B^{\ch(\vl)}, \icvest_{B^{\ch(\vl)}}(\ch(\vl)))$ to her children and plays the action $A^\vl$.

\textbf{Incentive Exploration Subroutine.} \Cref{algorithm:algorithm} is based on the payment exploration subroutine $\expsub$ which estimates $\icvs_b(\wl)$ in the range $[\licv_b, \hicv_b]$ and shrinks the gap $\hicv_b(\lambda)-\licv_b(\lambda)$ along the iterations $\lambda$. More precisely, consider any player $\vl$ and action $b \in \cA$. Thanks to \Cref{assumption:bounded_rewards}, $\expsub$ starts with an estimate $[\licv_b(0), \hicv_b(0)] = [0,1]$ that contains $\icvs_b(\wl)$ almost surely. For a batch $\lambda$ of $\lceil T^\alpha\rceil$ iterations, $\vl$ offers an incentive $(b, (\hicv_b(\lambda)+\licv_b(\lambda))/2)$ to $\wl$. At the end of the batch $\lambda$, the number of steps for which the incentive has been accepted allows to \textit{guess} whether $\icvs_b(\wl)>(\hicv_b(\lambda)+\licv_b(\lambda))/2$ or $\icvs_b(\wl) < (\hicv_b(\lambda)+\licv_b(\lambda))/2$. $\expsub$ runs for $K \lceil T^\alpha\rceil \lceil \log T^\beta\rceil$ steps until a precision $1/T^\beta$ is obtained on $\icvs_b(\wl)$ for any $b\in [K], \wl\in \ch(\vl)$ and then estimates
\begin{equation}\label{equation:estimated_incentives_definition}
    \icvest_b(\wl) = \hicv_b(\lceil \log T^\beta \rceil)+1/T^\beta + \acst \,\mB \, T^{-\extra} \eqsp,
\end{equation}
where the quantity $1/T^\beta$ ensures that $\icvest_b(\wl)$ is above $\icvs_b(\wl)$ with high probability while the term $\acst \,\mB \, T^{-\extra}$ is there to make action $b$ really better than the others and thus, force $\wl$ to play it. As it is shown formally in \Cref{proposition:estimated_incentives_good} from \Cref{appendix:omitted_proofs}, we have the following for the estimated payment
\begin{equation}\label{equation:good_precision_incentives}
    |\icvest_b(\wl)-\icvs_b(\wl)| = o(T) \eqsp.
\end{equation}
After the payment exploration phase, suppose that a player $\vl$ takes action $A_t^\vl$ at step $t \in [T]$ and offers an incentive $(B_t^{\ch(\vl)},\icvest_{B_t^{\ch(\vl)}}(\ch(\vl)))$ to her children. If they all best-respond, then $(A_t^\vl,A_t^{\ch(\vl)}) = (A_t^\vl,B_t^{\ch(\vl)})$, which means that we can define the quantity $Z_t(A_t^\vl,A_t^{\ch(\vl)})$ that is the reward observed by player $\vl$, taking into account the cost of the actions, following $Z_t^\vl(A_t^\vl,B_t^{\ch(\vl)}) = X_t^\vl(A_t^\vl,B_t^{\ch(\vl)}) - \sum\icvest_\wl(B_t^\wl)$, as well as the true mean of this quantity $\nu_t^\vl(A_t^\vl,B_t^{\ch(\vl)}) = \theta_t^\vl(A_t^\vl,B_t^{\ch(\vl)}) - \sum\icvest_\wl(B_t^\wl) $.

We can now state our most important technical lemma, which allows the nodes of the tree to play well enough as we go from the leaves to the nodes. If node $\vl$ runs $\alg$ with parameters $(\extra, \alpha, \beta)$ and her children satisfy \Cref{assumption:bounded_regret} with parameters $(\wait, \acst, \kappa, \zeta)$, we need to have
\begin{equation}\label{equation:condition_parameters}
    \beta/\alpha < 1- \kappa \eqsp,
\end{equation}
for the payment exploration subroutine to be successful (see \Cref{appendix:omitted_proofs}, \Cref{lemma:classification_search} for more details). As shown in \eqref{equation:estimated_incentives_definition}, $\mB \acst T^{-\extra}$ is the amount of extra payment added to each arm to ensure that the children accept the incentives for enough steps. The regret bounds imply both a term evolving as $\bigO(T^{\kappa+\eta})$ and a term evolving as $\bigO(T^{1-\eta})$, hence the need of choosing $\eta$ such that $0<\extra<1-\kappa$, which is satisfied with our choice of parameters in \eqref{equation:definition_good_parameters}.

\begin{restatable}{lemma}{repeatassumptionbis}\label{lemma:repeat_assumption}
Consider a player $\vl \in \trV$ running $\alg$ with parameters $(\extra, \alpha, \beta)$ such that any child $\wl \in \ch(\vl)$ satisfies \Cref{assumption:bounded_regret} with parameters $(\wait, \acst, \kappa, \zeta)$ and $\beta/\alpha<1-\kappa$. Assume that \Cref{assumption:bounded_rewards} holds. In that case, $\vl$ also satisfies \Cref{assumption:bounded_regret} with parameters
$(\wait+K\lceil T^\alpha \rceil \lceil \log T^\beta \rceil, 10 \sqrt{K^{\mB+1}\log(K^{\mB+1}T^3)}, \max\{1/2, \kappa+\extra, 1-\beta\}, \Tilde{\zeta})$ with $\Tilde{\zeta} = \alpha \zeta- \log(4 \mB K \log T)/\log T$.
\end{restatable}

\begin{restatable}{corollary}{boundedregretv}\label{corollary:boundedregretv}
    Consider a player $\vl \in \trV$ running $\alg$ such that any child $\wl \in \ch(\vl)$ satisfies \Cref{assumption:bounded_regret}. Then, we have that
    \begin{equation*}
        \E[\regret^\vl(T)] \leq o(T) \eqsp.
    \end{equation*}
\end{restatable}

\textbf{Tree Structure.} We now specifically study the game within the tree structure $\cT$. The hyperparameters $(\extra_d, \alpha_d,\beta_d)$ are chosen accordingly to the depth $d \in \iint{2}{D}$. In our setting for the next theorems, $(\wait_d, \acst_d, \kappa_d, \zeta_d, \alpha_d,\beta_d)$ are the layers' parameters and are common to any node $\vl$ that is located at depth $d$ in the tree. Assume that for any $d \in \iint{2}{D}$, a node $\vl$ located at depth $d$ runs $\alg$ with parameters
\begin{equation}\label{equation:definition_good_parameters}
\begin{aligned}
    & \extra_d = \frac{1}{2d(d-1)} \; , \; \alpha_d = \frac{(D+1)(d-1)}{D \, d} \; , \; \beta_d = \frac{1}{2d} \eqsp,
\end{aligned}
\end{equation}
where $(\alpha, \beta, \extra)$ are not defined for the leaves which run directly \Cref{algorithm:ucb} without \Cref{algorithm:search_subroutine}. \Cref{lemma:repeat_assumption} is the key ingredient to prove our regret bounds. An important feature is the extra payment $\acst \, \mB \, T^{-\eta}$ in $\icvest$, which creates a larger reward gap and forces the children to play the action recommended by the parent. 
While this extra payment largely increases the payment regret, it yields a better control of the deviation regret as it makes deviating from the recommended actions more costly for the children. Moreover, any children deviation yields feedback on different arms than the ones intended by the principal and thus leads to lost exploration steps. As a consequence, this extra payment is not only needed to control the deviation regret, but also the action regret.

\begin{restatable}{lemma}{everyoneassumption}\label{lemma:everyone_assumption}
Consider a large enough horizon so that $\log T \geq D^2 \log(4 \mB K \log T)$. Assume that any player at depth $d$ in the tree runs $\alg$ with the parameters defined in \eqref{equation:definition_good_parameters}. Then, any player at depth $d \in \iint{1}{D}$ satisfies \Cref{assumption:bounded_regret} with parameters $(d\, K \lceil T^\alpha \rceil \lceil \log T^\beta \rceil, \acst, 1-1/2d, \zeta_d)$, and we have $\zeta_d \geq 1/d$.
\end{restatable}

The assumption $\log T \geq D^2 \log(4 \mB K \log T)$ may seem strong. However, it only reflects the fact that the horizon needs to be large enough as compared to the size of the tree, which is natural since the opposite would not allow players to learn from others within the available amount of time.

\begin{restatable}{theorem}{mainconvergence}\label{theorem:main_convergence}
Consider the game running for a horizon $T$, large enough so that $\log T \geq D^2 \log(4 \mB K \log T)$. Assume that any player in the tree runs \Cref{algorithm:algorithm} with the parameters specified in \eqref{equation:definition_good_parameters}. In that case, for any player $\vl \in \trV$ at depth $d$ in the tree, we have that
\begin{equation*}
        \regret^\vl(T) = \Tilde{\bigO}(T^{1-\frac{1}{2d^2}})
        \eqsp,
\end{equation*}
with probability at least $1-T^{\zeta_d}$ for $\zeta_d\geq 1/d$, and $\Tilde{\bigO}$ hiding logarithmic factors.
\end{restatable}
In particular, \Cref{theorem:main_convergence} implies that we have
\begin{align*}
\E\parentheseDeux{\regret^\vl(T)} = o(T) \eqsp,
\end{align*}
and hence $\alg$ is actually a no-regret strategy for the players in the tree. This result is exciting because the \textit{orchestration} of the whole tree appears without the need of any kind of collaboration nor centralized leader. The single fact of allowing contracts makes the players behave optimally and maximize the social welfare, as it is formulated clearly in \Cref{corollary:final_convergence}, from a \textit{coalition} point of view. This result shows that, although being selfish on their own, the players can converge together towards the global equilibrium. This is accomplished through repeated transfers between the players, which helps avoiding globally suboptimal actions at any step. Note that the scaling of the regret depends on your depth in the tree: you have a lower regret when you interact with better-responding agents - i.e. lower ranked in the depth.

\begin{restatable}{corollary}{finalconvergence}\label{corollary:final_convergence}
    Suppose that the horizon is large enough so that $\log T \geq D^2 \log(4 \mB K \log T)$ and that any player $\vl \in \trV$ runs \Cref{algorithm:algorithm} with the parameters specified in \eqref{equation:definition_good_parameters}. In that case, we have that
\begin{align*}
    & \E\parentheseDeux{\sum_{t =1}^T \max_{(a^\vl) \in \cA^{|\trV|}} \sum_{\vl \in \trV}  \theta^\vl(a^\vl, a^{\ch(\vl)})- \theta^{\vl}(A_t^\vl, A_t^{\ch(\vl})} \\
    & \quad = o(T) \eqsp.
\end{align*}
\end{restatable}

As we show formally in \Cref{lemma:spne_outcome_convergence} in the appendix, this lemma and its proof allow to show that the profile of actions played by players in $\trV$ following $\alg$ converge to the profile of actions of players $\trV$ whose strategy are a subgame perfect Nash equilibrium.

\section{Final Remarks and Conclusion}

\textbf{Speed of convergence.} An interesting feature is the evolution of the regret's \textit{speed of convergence} along the layers -- the worst regret bound being for the root. \textit{The better your agent behaves, the lower is your regret} -- the optimal scenario being \textit{interactions with a best-responding} agent. Here, we showed that incentivization is still possible among a tree-structure of players.

\textbf{Why Decomposing the Regret?} \Cref{assumption:bounded_regret} could be stated in terms of the overall regret $\regret^\vl$ and it would still yield convergence guarantees, but with dramatic regret bounds, as we show here. 
Beyond enhancing clarity, decomposing the regret as we did in \Cref{lemma:decomposition_regret} is crucial to mitigate the dramatic propagation of regret terms.

\begin{restatable}{lemma}{lowerboundish}\label{lemma:lowerboundish}
    Following the same setup as described, consider a parent $\vl$ interacting with his child $\wl$. For any strategy of the parent $\vl$, there exists a game $(\theta^\vl, \theta^\wl)$ and a child's strategy that achieves a regret upper bound $\regret^\wl(T) \leq \bigO(T^\kappa)$ such that the parent's regret is lower bounded as $$\regret^\vl(T)\geq \Omega \, (T^{\frac{\kappa+1}{2}}) \eqsp.$$
\end{restatable}
This lemma demonstrates that regret propagation can be extremely severe in hierarchical structures. Considering \Cref{lemma:lowerboundish}, a straightforward inductive argument shows that a player $\vl$ at depth $d$ would suffer a regret on the order of $\regret^\vl \approx T^{1-\frac{1}{2^d}}$. Although our method — through a careful decomposition of payment, deviation, and action regret — mitigates this escalation, it is important to recognize that even small amounts of regret at the agent level can amplify dramatically as they propagate upward. This observation also justifies our conservative regret upper bounds, which may initially seem loose but capture the unavoidable propagation of errors in the tree.

\textbf{Why transfers.} One could ask why the transfers were introduced in this sequential multi-agent game. As shown in \citet{scheid2024learning} in the simpler two-players case, without transfers, one would necessarily have
\begin{align*}
    & \sum_{t =1}^T \max_{(a^\vl) \in \cA^{|\trV|}} \sum_{\vl \in \trV}  \theta^\vl(a^\vl, a^{\ch(\vl)})- \theta^{\vl}(A_t^\vl, A_t^{\ch(\vl}) \geq \Delta T \eqsp,
\end{align*}
for some reward gap $\Delta>0$. It is directly linked to the fact that without transfers, the utilities of the players are not aligned, which creates inefficiencies and results in a loss of social welfare. Without transfers in this multi-agent game, the global utility breaks down. Hence, \textit{contracting and incentives restore the efficiency of the game.}

\textbf{Experiments.} We illustrate our theoretical findings with experiments on a toy example in \Cref{figure:regret_comparison_non_contextual} where there is a depth $D=3$ and each node -- except the leaves -- has $3$ children. It makes a total of $13$ players in the game.

\begin{figure}[t]
  \centering
  \includegraphics[width=\columnwidth]{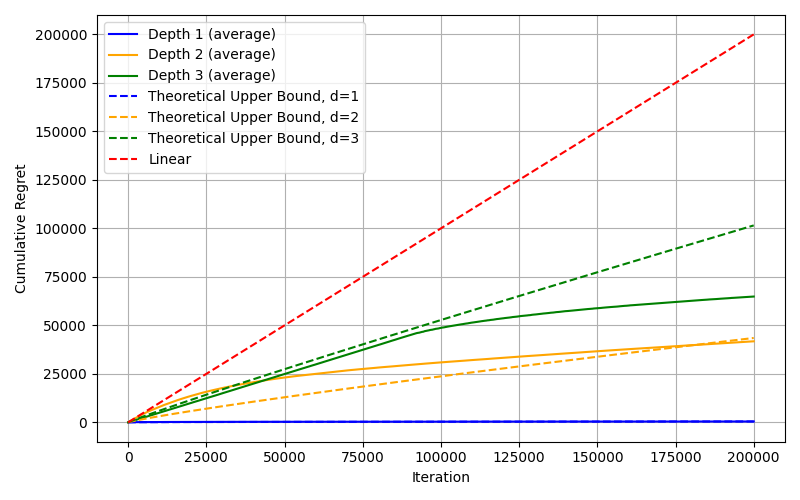}
  \caption{Cumulative regret for different algorithms on a $5$ arms instance.}
    \label{figure:regret_comparison_non_contextual}
\end{figure}

There are $K=5$ arms for each players and each reward for a user and an arm is a uniform draw in $[0,1]$. Additionally, the dotted line represents the dominant term in the upper bound, $T^{1-1/2d^2}$.

\paragraph{Conclusion.} We showed that allowing transfers between non-collaborative players allows no-regret learners to converge to the global optimum within a hierarchical principal-agent tree structure. Algorithmically, the first step for the players is to learn the optimal incentive needed to \textit{convince} their children to play some specific action. Then, a bandit subroutine can be used on the shifted instance. Interestingly, the players end up behaving as if they were collaborating. The straightforward extensions of this work would be similar studies with dynamic trees or additional context. It would also be interesting to relax to some extent the \textit{myopic} and \textit{regret-maximizing} assumptions on our agents.

\newpage

\section*{Acknowledgements}

Funded by the European Union (ERC, Ocean, 101071601). Views and opinions expressed are however those of the authors only and do not necessarily reflect those of the European Union or the European Research Council Executive Agency. Neither the European Union nor the granting authority can be held responsible for them.


\bibliography{sample}

\clearpage
\appendix
\thispagestyle{empty}

\onecolumn

\newpage
\appendix

\section{Extended Related Works}

The emergence of large models considered as economic agents \citep{shapira2024can, immorlica2024generative, deng2024llms} motivates the study of learning in repeated \textit{Stackelberg games} \citep{brown2024learning, morri2025learning}. The latter is challenging because the leader generally does not know the follower's preferences \citep{arunachaleswaran2024learning}. From a global perspective, welfare maximization with learning agents at stake is a complex task \citep{anagnostides2024barriers}: a central question is whether players can converge to equilibrium (which has been explored in \citet{zhang2023steering}) or if it is possible to steer no-regret agents towards desired outcomes \citep{widmer2025steering}. However, the setup of \citet{zhang2023steering} assumes the presence of a mediator interacting with multiple agents, thereby overlooking the intricate interactions and decentralized nature of the problem. Incentives have proved to be efficient to align players with initially misaligned goals \citep{saig2024incentivizing, maheshwari2024adaptive, freund2025fair}, and framing learning problems as games in general can help to design algorithms for interacting environments  \citep{isa2025learning}.

\section{Detailed Explanations on the Algorithms}\label{appendix:algorithmic_details}

\textbf{Payment Exploration Phase for Player $\vl \in \trV$.} A first step is to explain how works the payment exploration subroutine \Cref{algorithm:search_subroutine}, which roughly corresponds to a batched binary search operated by the parent node $\vl \in \trV$ on all of her children $\wl \in \ch(\vl)$. As presented in \Cref{lemma:optimal_incentives}, the true rewards (which takes into account the transfers to the children) of the actions can be defined by induction. The goal of a player $\vl \in \trV$ is to learn first the optimal transfer $\icvs_b(\wl)$ for any $b \in \cA, \wl \in \ch(\vl)$ which then gives the optimal contract $(b,\icvs_b(\wl))$. For that purpose, $\vl$ first waits $\wait = \max_{\wl \in \ch(\vl)} \wait_\wl$ steps, so that all her children almost best-respond on their own bandit instances with their own children nodes. After these steps, the player $\vl$'s payment exploration phase starts. The binary search strategy is the same across all the children. More precisely, for any $b \in \cA$, we define
\begin{equation}\label{equation:psi}
    \Psi_b^\vl = \wait + (b-1) \, \lceil \log T^\beta \rceil \, \lceil T^\alpha \rceil \eqsp,
\end{equation}
which is the step at which $\vl$ starts the binary search procedure to estimate arm $b$'s associated optimal transfer $\icvs_b(\wl)$ - for any $\wl \in \ch(\vl)$. For any binary search step $\lambda \in \iint{1}{\lceil \log T^\beta \rceil}$, we also define
\begin{equation}\label{equation:k_b_d}
    k_{b,\lambda}^\vl = \Psi_b^\vl + (\lambda-1)\, \lceil T^\alpha \rceil \eqsp,
\end{equation}
which corresponds to the step after which $\vl$ starts the $\lambda$-th batch of binary search on her children's arm $b$. Now, we can define
\begin{equation}\label{equation:Tnot}
    \Tnot_{b,\lambda}(\wl) = \Card\{ t \in \{k_{b,\lambda}^\vl+1, \ldots, k_{b,\lambda}^\vl+\Texp\} \; \text{ such that }\; A_t^\wl \ne b\} \eqsp,
\end{equation}
which corresponds to the number of steps during the payment exploration phase on arm $b$, player $\wl$ when the child did not accept the transfer. We also define our "nice" event $\evgoodw$ for any player $\wl \in \trV$ on a window of steps $\iint{s+1}{s+t}$ as
\begin{align}
\evgoodw(\iint{s+1}{s+t}, \acst, \kappa) & = \left\{\sum_{l=s+1}^{s+t} \max_{a \in \cA} \{\reww_a + \1_{B_l^\wl}(a)\, \icv_l(\wl)\} \right. \notag \\
& \quad \left. - \left\{ \mu^\wl(A_l^\wl, B_l^{\ch(\wl)}) +\1_{B_l^\wl}(A_l^\wl)\, \icv_l(\wl)\right\} \leq \acst t^{\kappa} \right\} \eqsp,
\end{align}
and any node $\wl \in \ch(\vl)$ which satisfies \Cref{assumption:bounded_regret} with parameters $(\wait, \acst, \kappa, \zeta)$, also satisfies $\evgoodw([s+1,\ldots, s+t], \acst, \kappa)$ for any $s>\wait$.

To clarify, \Cref{assumption:bounded_regret} ensures that the children actions $(A^{\ch(\vl)})$ do not deviate too much from the recommended action $(B^{\ch(\vl)})$. With mean-based algorithms such as \texttt{UCB}, the contract $(B^\vl_t,\icv_t(\vl))$ received from $\pa(\vl)$ at each step $t$ can simply be added to the score of each arm to be taken into account. Based on this assumption, \Cref{lemma:repeat_assumption} shows that if $\vl$ uses the algorithm $\alg$ and all her children satisfy \Cref{assumption:bounded_regret}, $\vl$ actually satisfies the assumption in turn.

Once the leaves satisfy \Cref{assumption:bounded_regret}, their parents -- if they use $\alg$ -- also satisfy \Cref{assumption:bounded_regret} with new constants, again by \Cref{lemma:repeat_assumption}. By induction, \Cref{assumption:bounded_regret} is then satisfied by all the nodes up to the root if everyone plays $\alg$, where the convergence rates are detailed in \Cref{lemma:everyone_assumption}. In that way, we show that if all the players use $\alg$, then they converge to the optimum as a whole.

Before moving to the results, we first describe more precisely the different subroutines used in $\alg$. Assuming that the children $\ch(\vl)$ of a node $\vl$ satisfy \Cref{assumption:bounded_regret} with parameters $(\acst_\wl, \wait_\wl, \kappa_\wl, \zeta_\wl)$, we can define
\begin{equation}\label{equation:domination_parameters}
\begin{aligned}
    &\acst = \max_{\wl \in \ch(\vl)}\acst_\wl \;\; , \;\; \wait =\max_{\wl \in \ch(\vl)} \wait_\wl \eqsp, \kappa = \max_{\wl \in \ch(\vl)} \kappa_\wl  \;\; ,  \;\; \zeta = \min_{\vl \in \ch(\wl)} \zeta_\wl \eqsp.
\end{aligned}
\end{equation}

\begin{algorithm}[tb]
   \caption{$\alg$ (player $\vl$).}
   \label{algorithm:algorithm}
\begin{algorithmic}
   \STATE {\bfseries Input:} Set of action $\cA$, horizon $T$, subroutine $\sub^\wl$ and parameters $\acst_\wl$, $\wait_\wl$, $\kappa_\wl$, $\zeta_\wl$ for any $\wl \in \ch(\vl)$, parameters $\alpha, \beta, \extra$
   \STATE Compute $\cH^\vl_0 = \varnothing, \wait = \max_\wl \wait_\wl,\acst = \max_\wl \acst_\wl, \kappa = \max_\wl \kappa_\wl, \zeta = \min_\wl \zeta_\wl$
   \STATE Pick action $B \in \cA$ at random.
   \FOR{$t=1, \ldots, \wait$}
    \STATE Offer contract $(B,0)$ to any child $\wl \in \ch(\vl)$ 
    \STATE Play action $A_t^\vl \sim \mathrm{Unif}(K)$.
   \ENDFOR
   \FOR{any agent $\wl \in \ch(\vl)$}
   \FOR{any action $b \in \cA$}
   \STATE $\licv, \hicv = \expsub(T,\wl, \sub^\wl, b, \wait, \acst, \kappa, \alpha, \beta, \extra)$
   \STATE Compute $\icvest_b(\wl) = \hicv_b(\lceil \log T^\beta \rceil)+1/T^\beta + \acst \,\mB \, T^{-\extra}$
   \STATE \textcolor{blue}{\# See \Cref{algorithm:search_subroutine}}
   \ENDFOR
   \ENDFOR
   \FOR{$t = \wait+K \lceil T^\alpha \rceil \lceil \log T^\beta\rceil+1, \ldots, T$}
   \STATE Get set of actions from the $K^{\mB+1}$ bandit instance: \\
   $(A_t^\vl, B_t^{\ch(\wl)}) = \bandsub(U_t,\cH_{t-1}^\vl, B^\vl_t, \icv_t(\vl))$.
   \STATE Play action $A_t^\vl$, offer contract $\cC_t^\wl = (B_t^\wl,\icvest_{B_t^\wl}(\wl))$ to any player $\wl \in \ch(\vl)$.
   \STATE Observe $A_t^\wl = \sub^\wl(U_t^\wl, \cH_{t-1}^\wl, B_t^\wl, \icv_t(\wl))$ for any child $\wl \in \ch(\vl)$
   \STATE Update history $\cH_t^\vl$ and history $\cH_t^\wl$ for any $\wl \in \ch(\vl)$
   \ENDFOR
\end{algorithmic}
\end{algorithm}

\begin{algorithm}[tb]
   \caption{$\expsub$ (player $\vl$).}
   \label{algorithm:search_subroutine}
\begin{algorithmic}
   {\STATE \bfseries Input:} $T,\wl, \sub^\wl, b, \wait, \acst, \kappa, \alpha, \beta, \extra$.
   \STATE Initialize $\licv_b(0)=0, \hicv_b(0) = 1, \Texp = \lceil T^\alpha \rceil$ 
    \FOR{$\lambda = 0, \ldots, \lceil \log (T^\beta) \rceil-1$}
        \STATE Compute $\icvmid_b(\lambda) = (\hicv_b(\lambda) + \licv_b(\lambda))/2$, $\Tnot_b=0$.
        \FOR{$t = \lambda \Texp+1, \ldots, \lambda \Texp + \Texp$}
            \STATE Propose contract $(b, \icvmid_b(\lambda))$ to player $\wl$
            \STATE $A^\wl_t = \sub^\wl(\cU_{t+1}^\wl,\cH_t^\wl,b, \icvmid_b(\lambda))$
            \IF{$A_t^\wl \ne b$:}
            \STATE $\Tnot_b +=1$
            \ENDIF
            \STATE Update player's $\wl$ history $\cH^\wl_{t+1}$ following the definition
    \ENDFOR
    \IF{$\acst \Texp^{\kappa +\beta/\alpha} < \Tnot_b < \Texp -\acst\Texp^{\kappa+\beta/\alpha}$}
    \STATE Return $\licv_b(\lambda), \hicv_b(\lambda)$.
    \ELSIF{ $\Tnot_b \leq \Texp - \acst \Texp^{\kappa+\beta/\alpha}$}
        \STATE $\hicv_b(\lambda) = \icvmid_b(\lambda)+1/T^\beta$ and update history $\cH^\vl_{t}$.
        \ELSE
        \STATE $\; \licv_b(\lambda) = \icvmid_b(\lambda)-1/T^\beta$ and update history $\cH^\vl_{t}$.
    \ENDIF
    \ENDFOR
\end{algorithmic}
\end{algorithm}

\textbf{\texttt{UCB} Subroutine.} We present the player's \texttt{UCB} subroutine. As explained above, once the optimal contracts are estimated by \Cref{algorithm:search_subroutine}, $\alg$ uses \Cref{algorithm:ucb} to be no-regret on her \textit{shifted} bandit instance (the reward of each arm is shifted by the contract that needs to be paid so the children actually play the right action). As shown in \citet{scheid2024learning}, the policy \Cref{algorithm:ucb} satisfies \Cref{assumption:bounded_regret} with $\kappa=1/2, \zeta=2$. We make use of this for the proofs in \Cref{appendix:omitted_proofs}. Since the leaves of the tree directly use \Cref{algorithm:ucb} without the need of offering contracts, they are guaranteed to satisfy \Cref{assumption:bounded_regret}.

\begin{algorithm}[!ht]
\caption{\texttt{UCB}'s subroutine - Player $\vl$}\label{algorithm:ucb}
\begin{algorithmic}[1]
    \STATE {\bfseries Input:} Set of arms $\cA$, horizon $T$, $\extra$, estimated incentives $(\icvest_b(\wl))_{\wl \in \ch(\vl), b \in \cA}$
    \STATE {\bfseries Initialize:} For any combination of arms $(a,b^{\ch(\vl)}) \in \cA^{\mB+1}$, set $\hat{\nu}_{a,b^{\ch(\vl)}} = 0, \, T_a = 0$
    \FOR{$1 \leq t \leq K^{\mB+1}$:}
        \STATE Pull an arm $A_t,B_t^{\ch(\vl)}$ that has not been pulled yet.
        \STATE Update $\hat{\nu}_{A_t,B_t^{\ch(\vl)}} = X_{A_t,B_t^{\ch(\vl)}}(t) -\sum_{\wl \in \ch(\vl)} \icvest_{B_t^{\wl}}(\wl), \, T_{A_t,B_t^{\ch(\vl)}}(t)= 1$
    \ENDFOR
    \FOR{$t \geq K^{\mB+1}+1$}
        \STATE Observe the contract $(B_t^\vl, \icv_t(\vl))$.
        \STATE Pick $(A_t^\vl,B_t^{\ch(\vl)}) \in \argmax_{a,b^{\ch(\vl)} \in \cA^{\mB+1}} \left\{\hat{\nu}_{a,b^{\ch(\vl)}}(t-1) + 2\sqrt{\frac{\log (K^{\mB+1} \, T^3)}{T_{a,b^{\ch(\vl)}}(t-1)}} + \1_{B_t^\vl}(a)\icv_t(\vl) \right\}$
        \IF{$A_t^{\ch(\vl)}=B_t^{\ch(\vl)}$}
        \STATE Update $T_{A_t^\vl,A_t^{\ch(\vl)}}(t)=T_{\A_t^\vl,A_t^{\ch(\vl)}}(t-1)+1, \, \hat{\nu}_{A_t^\vl,A_t^{\ch(\vl)}}(t) = \frac{1}{T_{A_t^\vl,A_t^{\ch(\vl)}}(t)}(T_{A_t^\vl,A_t^{\ch(\vl)}}(t-1)\hat{\nu}_{A_t^\vl,A_t^{\ch(\vl)}}(t-1) + X_{A_t^\vl,A_t^{\ch(\vl)}}(t)-\sum_{\wl \in \ch(\vl)} \icvest_{B_t^{\wl}}(\wl))$
        \ENDIF
    \ENDFOR
\end{algorithmic}
\end{algorithm}

\section{Proofs}\label{appendix:omitted_proofs}

\begin{proof}[Proof of \Cref{lemma:optimal_incentives}]
    We conduct the proof by induction starting from the first layer $d=1$ contracting with the leaves and ascending to the root.

    \textbf{Initialization.} The optimal utility any leaf $\vl$ can extract from the arms is $\theta^\vl(a)$ for any action $a\in \cA$. Therefore, we have that for $\vl$, the average optimal utility is defined as
    \begin{equation*}
        \mu^\vl(a) = \theta^\vl(a) \; \text{  and  } \; \rewv(a) = \theta^\vl(a) \eqsp,
    \end{equation*}
    since $\ch(\vl) = \varnothing$ and hence $\vl$'s contracts are empty. As we show it in the main text, we also have that for any node $\ul$ in the second layer (in interaction with children $\vl \in \ch(\ul)$ being leaves), we have that
    \begin{equation*}
        \icvs_{b^\vl}(\vl) = \max_{b \in \cA} \rewv(b) - \rewv(b^\vl) = \max_{b \in \cA} \theta^\vl(b) - \theta^\vl(b^\vl) \eqsp,
    \end{equation*}
    and this quantity is well-defined and unique. Considering $\ul=\pa(\vl)$, it leads to
    \begin{equation*}
        \mu^\ul(a,b^{\ch(\ul)}) = \theta^\ul(a,b^{\ch(\ul)}) - \sum_{\vl \in \ch(\ul)} \icvs_{b^\vl}(\vl) \eqsp,
    \end{equation*}
    which is well defined too. Finally, for any $a \in \cA, \rewu(a)$ is well-defined and unique as a maximum over a finite number of scalars. Therefore, our initialization holds.

    \textbf{Induction step.} Consider a player $\vl \in \trV$ that is not a leaf. We consider a single child $\wl \in \ch(\vl)$ to show our result. Our induction assumption gives us that $\reww(b)$ is well-defined for any $b \in \cA$. For any $b^\wl \in \cA$, picking $b^\wl$ can lead in hindsight to an optimal average utility of $\reww(b^\wl)$ by definition. Since $\icvs_{b^\wl}(\wl)$ is defined as the smallest transfer such that $b^\wl$ is the best action $\wl$ could play in hindsight, we necessarily have that for any $b \in \cA$
    \begin{equation}\label{equation:fsrgf}
        \1_{b^\wl}(b^\wl)\, \icvs_{b^\wl}(\wl) + \reww(b^\wl) \geq \1_{b^\wl}(b)\, \icvs_{b^\wl}(\wl) + \reww(b) \eqsp,
    \end{equation}
    which gives that necessarily
    \begin{align*}
        \icvs_{b^\wl}(\wl) \geq \max_{\b \in \cA} \reww(b) - \reww(b^\wl) \eqsp.
    \end{align*}
    As long as $\icv_{b^\wl}(\wl) >\max_{\b \in \cA} \reww(b) - \reww(b^\wl)$, \eqref{equation:fsrgf} is satisfied with a strict inequality: which makes $b^\wl$ the best possible action in hindsight. Hence, taking the infimum of $\icvs_{b^\wl}(\wl)$ gives that
    \begin{align*}
        \icvs_{b^\wl}(\wl) = \max_{\b \in \cA} \reww(b) - \reww(b^\wl) \eqsp,
    \end{align*}
    hence the first part of the result. For any couple $(a,b^{\ch(\vl)}) \in \cA \times \cA^\mB$, if the players $\ch(\vl)$ are fully rational and have a full knowledge of the game, then offering to any player $\wl \in \ch(\vl)$ the contract $\cC^\wl = (b^\wl, \icv_{b^{\wl}}(\wl))$ such that $\icv_{b^{\wl}}(\wl)\geq \icvs_{b^{\wl}}(\wl)$ ensures that $A^{\wl} = b^\wl$. This leads to the average utility
    \begin{align*}
        \E\parentheseDeux{\cU^\vl(a,\cC^{\ch(\vl)})} = \theta^\wl(a,A^{\ch(\vl)}) - \sum_{\wl \in \ch(\vl)} \1_{b^\wl}(A^\wl) \, \icv_{b^\wl}(\wl) = \theta^\wl(a,b^{\ch(\vl)}) - \sum_{\wl \in \ch(\vl)} \icv_{b^\wl}(\wl) \eqsp,
    \end{align*}
    and finally taking the maximum over the transfers that only need to satisfy $\icv_{b^{\wl}}(\wl)\geq \icvs_{b^{\wl}}(\wl)$ gives that
    \begin{align*}
        \mu^\vl(a,b^{\ch(\vl)}) = \theta^\wl(a,b^{\ch(\vl)}) - \sum_{\wl \in \ch(\vl)} \icvs_{b^\wl}(\wl) \eqsp.
    \end{align*}
    Since $(\mu^\vl(a,b^{\ch(\vl)}))_{a \in \cA, b^{\ch(\vl)}\in \cA^\mB}$ are finite, unique and well-defined quantities, we have that
    \begin{equation*}
        \rewv(a) = \max_{b^{\ch(\vl)}\in \cA^\mB} \mu^\vl(a,b^{\ch(\vl)})
    \end{equation*}
    is also well-defined and unique. Therefore, our induction holds, hence the result.

For the second part of the proof, we first write $\trV_d$ for the set of nodes $\vl \in \trV$ such that $\vl$ is located at depth $d$ and we consider the quantity
    \begin{align*}
        \sum_{\vl \in \trV} \max_{(a^\vl) \in \cA^{\Card(\trV)}}
        \rewv(a^\vl) & = \sum_{\vl \in \trV_1}
        \max_{a^\vl}\theta^\vl(a^\vl) + 
        \sum_{\vl \in \trV_2}
        \max_{a^\vl,a^{\ch(\vl)}} \left\{\theta^\vl(a^\vl,a^{\ch(\vl)}) -\sum_{\wl \in \ch(\vl)}\icvs_{a^{\wl}}(\wl)\right\} + 
        \ldots \\
        & \quad + \sum_{\vl \in \trV_D}\max_{a^\vl,a^{\ch(\vl)}} \left\{\theta^\vl(a,a^{\ch(\vl)}) -\sum_{\wl \in \ch(\vl)}\icvs_{a^{\wl}}(\wl)\right\} \\
        & = \sum_{\vl \in \trV_1}
        \max_{a^\vl}\theta^\vl(a^\vl) \\
        & \quad + \sum_{\vl \in \trV_2}
        \max_{a^\vl,a^{\ch(\vl)}} \left\{\theta^\vl(a^\vl,a^{\ch(\vl)}) -\sum_{\wl \in \ch(\vl)}(\max_{b \in \cA}\reww(b)-\reww(a^\wl)) \right\}\\
        & \quad + \ldots  + \sum_{\vl \in \trV_D}\max_{a^\vl,a^{\ch(\vl)}} \left\{\theta^\vl(a,a^{\ch(\vl)}) -\sum_{\wl \in \ch(\vl)}(\max_{b \in \cA}\reww(b)-\reww(a^\wl)) \right\} \\
        & = \sum_{\vl \in \trV_1}
        \max_{a^\vl}\theta^\vl(a^\vl) - \sum_{\vl \in \trV_1}
        \max_{a^\vl}\theta^\vl(a^\vl) \\
        & \quad + \sum_{\vl \in \trV_2}
        \max_{a^\vl,a^{\ch(\vl)}} \left\{\theta^\vl(a^\vl,a^{\ch(\vl)}) +\sum_{\wl \in \ch(\vl)}\theta^\wl(a^\wl) \right\} \\
        & \quad + \ldots  + \sum_{\vl \in \trV_D}\max_{a^\vl,a^{\ch(\vl)}} \left\{\theta^\vl(a,a^{\ch(\vl)}) -\sum_{\wl \in \ch(\vl)}(\max_{b \in \cA}\reww(b)-\reww(a^\wl)) \right\} \\
        & =\max_{a^\vl,a^{\ch(\vl)}} \left\{ \sum_{\vl \in \trV_1\cup \trV_2}\theta^\vl(a^\vl,a^{\ch(\vl)})\right\} +\ldots \\
        & \quad + \sum_{\vl \in \trV_D}\max_{a^\vl,a^{\ch(\vl)}} \left\{\theta^\vl(a,a^{\ch(\vl)}) -\sum_{\wl \in \ch(\vl)}(\max_{b \in \cA}\reww(b)-\reww(a^\wl)) \right\} \eqsp,
    \end{align*}
    where by convention $\ch(\vl)=\varnothing$ for any $\vl \in \trV_1$. Applying the same technique iteratively finally gives that
    \begin{align*}
         \sum_{\vl \in \trV} \max_{(a^\vl) \in \cA^{\Card(\trV)}}
        \rewv(a^\vl) =  \max_{(a^\vl) \in \cA^{\Card(\trV)}} \sum_{\vl \in \trV} 
        \theta(a^\vl,a^{\ch(\vl)}) \eqsp.
    \end{align*}
\end{proof}

Recall our main assumption \Cref{assumption:bounded_regret} that is fundamental for the convergence of the algorithms. A large part of the proof deals with this assumption and how it can be satisfied.

\begin{restatable}{lemma}{classificationsearch}\label{lemma:classification_search}
    Consider a player $\vl \in \trV$ and a child $\wl \in \ch(\vl)$ that satisfies \Cref{assumption:bounded_regret} with parameters $(\wait, \acst, \kappa, \zeta)$. Assume that $\vl$ runs the $\lambda$-th batch of binary search on player $\wl$'s arm $b \in \cA$: for any $t \in \iint{k_{b,\lambda}^\vl+1}{k_{b,\lambda}^\vl+\lceil T^\alpha \rceil}, (B^\wl_t, \icv_t(\wl)) = (b, \icvmid_b(\lambda))$. Let $\beta \in (0,1)$ such that $\beta<\alpha (1-\kappa)$ and recall the definition of $\Tnot_{b,\lambda}(w)$ from \eqref{equation:Tnot}. Conditionally on the event $\evgoodw(\iint{k_{b,\lambda}^\vl+1}{k_{b,\lambda}^\vl+\lceil T^\alpha\rceil}, \acst \Texp^\kappa)$, writing $\Texp = \lceil T^\alpha \rceil$, we have that
    \begin{itemize}
        \item If $\Tnot_{b,\lambda} < \Texp-\acst \Texp^{\kappa +\beta/\alpha}$, then $\icvs_b(\wl)<\icvmid_b(d) + 1/T^\beta$.
        \item If $\Tnot_{b,\lambda} > \acst \Texp^{\kappa +\beta/\alpha}$; then $\icvs_b(\wl)>\icvmid_b(\lambda) - 1/T^\beta$.
    \end{itemize}
    Consequently, with probability at least $1-2\,T^{-\alpha \zeta}$, if $\acst \Texp^{\kappa+\beta/\alpha} < \Tnot_{b,\lambda} < \Texp-\acst \Texp^{\kappa + \beta/\alpha}$, then $|\icvs_b(\wl)-\icvmid_b(\lambda)| \leq 1/T^\beta$.
\end{restatable}

\begin{proof}[Proof of \Cref{lemma:classification_search}]
    The whole proof is done conditionally on the event $\evgoodw(\iint{k_{b,\lambda}^\vl+1}{k_{b,\lambda}^\vl+\lceil T^\alpha \rceil}, \acst \Texp^\kappa)$, which holds with probability at least $1-T^{-\alpha \zeta}$, since $\wl$ satisfies \Cref{assumption:bounded_regret} with parameters $(\wait, \acst, \kappa, \zeta)$. For the ease of notation, we define $\icv_b = \icvmid_b(\wl)$. First suppose that $\icv_b \geq \icvs_b(\wl)+1/T^\beta$. By definition of the optimal transfer, we have that
    \begin{align*}
        \1_b(b) \icv_b + \reww(b) &\geq \icvs_b(\wl) + \reww(b) + 1/T^\beta \\
        & = \max_{b' \in \cA}\{ \reww(b')+\1_{b}(b')\icv_b \}+1/T^\beta \\
        & \geq \mu^\vl(b', B^{\ch(\wl)}) + \1_{b}(b') \icv_b +1/T^\beta\eqsp,
    \end{align*}
    for any $b', B^{\ch(\wl)} \in \cA\times \cA^\mB$. Hence $b$ is the best arm $\wl$ could play during $\vl$'s payment exploration batch $\{k_{b,\lambda}^\vl+1, \ldots, k_{b,\lambda}^\vl + \Texp\}$. Since $\wl$ satisfies \Cref{assumption:bounded_regret} with parameters $(\wait, \acst, \kappa, \zeta)$, we have that
    \begin{align*}
        \acst \Texp^\kappa & \geq \Tnot_{b,\lambda} \, \cdot (\icv_b+\reww(b)-\max_{b' \in \cA} \reww(b')) \\
        & \geq \Tnot_{b,\lambda}  \, \cdot (\icv_b-\icvs_b(\wl)) \\
        & \geq  \Tnot_{b,\lambda}/T^\beta \eqsp,
    \end{align*} 
    and since $\Texp\geq T^\alpha$, we obtain: $\Tnot_{b,\lambda}\leq \acst \Texp^{\kappa+\beta/\alpha}$. The contrapositive gives us that if $\Tnot_{b,\lambda}> \acst \Texp^{\kappa+\beta/\alpha}$ during the sequence $\{k_{b,\lambda}^\vl+1, \ldots, k_{b,\lambda}^\vl+\Texp\}$, then with probability at least $1-1/T^{\alpha \zeta}$, we have $\icv_b<\icvs_b(\wl)+1/T^\beta$.
    
    The exact same computation starting with the case $\icv_b \leq \icvs_b(\wl)-1/T^\beta$ gives that if $\Tnot_{b,\lambda}^\vl<\Texp-\acst \Texp^{\kappa+\beta/\alpha}$, then with probability at least $1-T^{-\alpha \zeta}$, $\icv_b>\icvs_b(\wl)-1/T^\beta$.

    Finally, considering the case $\acst \Texp^{\kappa+\beta/\alpha}<\Tnot_{b,\lambda} < \Texp^{\kappa+\beta/\alpha}-\acst \Texp^{\kappa+\beta/\alpha}$ concludes the proof.
\end{proof}

\begin{restatable}{lemma}{boundedincentives}\label{lemma:bounded_incentives}
    Assume that we run \Cref{algorithm:search_subroutine} and consider some payment exploration batch $\lambda \in \lceil \log T^\beta \rceil$ run on arm $b \in \cA$. Then $0 \leq \licv_b(\lambda) \leq \icvmid_b(\lambda) \leq \hicv_b(\lambda)  \leq 1$.
\end{restatable}

\begin{proof}[Proof of \Cref{lemma:bounded_incentives}]
    We prove the result by induction. Considering some action $b\in \cA$, the initialization of the proof holds thanks to the initialization of our algorithm: $\licv_a(0) = 0, \hicv_a(0)=1$ and therefore $\icvmid_a(0)\in [\licv_a(0), \hicv_a(0)]$. We now consider that $\lambda$ payment exploration batches have been run on $b$. Suppose that we run an additional batch on action $a$. We have
    \begin{equation}
    \label{lemBIbound}
    \icvmid_b(\lambda+1) = \frac{\hicv_b(\lambda)+\licv_b(\lambda)}{2} \text{ which gives } \icvmid_b(\lambda+1) \in [\licv_b(\lambda), \hicv_b(\lambda)] \eqsp.
    \end{equation}
    After this iteration of binary search, we either have $\hicv_b(\lambda+1) = \icvmid_b(\lambda+1)+1/T^\beta$ and $\licv_b(\lambda+1)=\licv_b(\lambda)$ or $\hicv_b(\lambda+1)=\icvmid_b(\lambda)$ and $\licv_b(\lambda+1)=\icvmid_b(\lambda+1)-1/T^\beta$. Therefore, we still have $0\leq \licv_b(\lambda+1)\leq \icvmid_b(\lambda+1) \leq \hicv_b(\lambda+1)\leq 1$, hence the result for any $\lambda \in \lceil \log T^\beta \rceil$ by induction.
\end{proof}

\begin{restatable}{lemma}{intersectionboundedregret}\label{lemma:intersection_bounded_regret}
    Consider a player $\vl \in \trV$ who runs $\Lambda$ batches of binary search on action $b \in \cA$ for all her children $\ch(\vl)$. Assume that any child $\wl \in \ch(\vl)$ satisfies \Cref{assumption:bounded_regret} with parameters $(\wait_\wl, \acst_\wl,\kappa_\wl, \zeta_\wl)$ such that for any $\wl \in \ch(\vl)$, we have $\wait_\wl \leq \wait, \acst_\wl \leq \wait,\kappa_\wl \leq \wait, \zeta_\wl \geq \zeta$. Then, we have that
    \begin{align*}
        \P\left(\bigcap_{\wl \in \ch(\vl)} \bigcap_{\lambda \in [\Lambda]} \evgoodw(\iint{k_{b,\lambda}^\vl+1}{k_{b,\lambda}^\vl+\lceil T^\alpha \rceil},\acst, \kappa) \right) \geq 1-\mB \, \Lambda /T^{\alpha \zeta} \eqsp,
    \end{align*}
    which can be rewritten as: the following event
    \begin{align*}
    & \text{For any } \wl \in \ch(\vl), \lambda \in [\Lambda], b \in \cA, \\
    &\quad \sum_{l=k_{b,\lambda}^\vl+1}^{k_{b,\lambda}^\vl+\lceil T^\alpha \rceil} \left\{\max_{b \in \cA} \{\rewv_b + \1_{B_t^\vl}(b)\, \icv_t(\vl)\} - 
    \left( \mu^\vl(A_t^\vl, B_t^{\ch(\vl)}) +\1_{B_t^\vl}(A_t^\vl)\, \icv_t(\vl) \right)\right\} \leq \acst \lceil T^\alpha \rceil^\kappa ,
\end{align*}
holds with probability at least $1-\mB \, \Lambda /T^{\alpha \zeta}$.
\end{restatable}

\begin{proof}[Proof of \Cref{lemma:intersection_bounded_regret}]
We do the proof considering the opposite event and applying an union bound. We have that
\begin{align*}
& \P\left(\bigcup_{\wl \in \ch(\vl)} \bigcup_{\lambda \in [\Lambda]} \left\{ \sum_{l=k_{b,\lambda}^\vl+1}^{k_{b,\lambda}^\vl+\lceil T^\alpha \rceil} \max_{b \in \cA} \{\rewv_b + \1_{B_l^\vl}(b)\, \icv_l(\vl)\} \right. \right. \notag \\
& \quad \quad \left. \left. - \sum_{l=k_{b,\lambda}^\vl+1}^{k_{b,\lambda}^\vl+\lceil T^\alpha \rceil}\left\{ \mu^\vl(A_l^\vl, B_l^{\ch(\vl)}) +\1_{B_l^\vl}(A_l^\vl)\, \icv_l(\vl)\right\} \right\} > \acst \lceil T^\alpha \rceil^\kappa \right) \\
& \quad \leq \sum_{\wl \in \ch(\vl)} \sum_{\lambda \in [\Lambda]}\P\left(\left\{\sum_{l=k_{b,\lambda}^\vl+1}^{k_{b,\lambda}^\vl+\lceil T^\alpha \rceil} \max_{b \in \cA} \{\rewv_b + \1_{B_l^\vl}(b)\, \icv_l(\vl)\} \right. \right. \notag \\
& \quad \quad \left. \left. - \sum_{l=k_{b,\lambda}^\vl+1}^{k_{b,\lambda}^\vl+\lceil T^\alpha \rceil}\left\{ \mu^\vl(A_l^\vl, B_l^{\ch(\vl)}) +\1_{B_l^\vl}(A_l^\vl)\, \icv_l(\vl)\right\} > \acst \lceil T^\alpha \rceil^\kappa \right\} \right) \\
    & \quad \leq \sum_{\wl \in \ch(\vl)} \sum_{\lambda \in [\Lambda]} 1/\lceil T^\alpha \rceil^\zeta \\
    & \quad \leq \mB \, \Lambda /T^{\alpha \zeta} \eqsp,
\end{align*}
where the first inequality holds with two union bounds, the second inequality holds because any $\wl \in \ch(\vl)$ satisfy \Cref{assumption:bounded_regret} and the last inequality holds because $\Card(\ch(\vl)) = \mB$. Considering the opposite event, we obtain that for any $\wl \in \ch(\vl)$, we have that the event
\begin{equation*}
    \left\{ \bigcap_{\wl \in \ch(\vl)} \bigcap_{\lambda \in [\Lambda]} \evgoodw(\iint{k_{b,\lambda}^\vl+1}{k_{b,\lambda}^\vl+\lceil T^\alpha \rceil},\acst, \kappa) \right\}
\end{equation*}
holds with probability at least $1-\mB \Lambda/T^{\alpha \zeta}$. Finally, the definition of $\evgoodw$ gives the first line of the lemma.
\end{proof}

\begin{restatable}{lemma}{precisionincentives}\label{lemma:precision_incentives}
Consider a player $\vl \in \trV$ who runs $\Lambda$ batches of binary search on action $b \in \cA$ for all his children $\ch(\vl)$. Assume that any child $\wl \in \ch(\vl)$ satisfies \Cref{assumption:bounded_regret} with parameters $(\wait_\wl, \acst_\wl,\kappa_\wl, \zeta_\wl)$ such that for any $\wl \in \ch(\vl)$, we have $\wait_\wl \leq \wait, \acst_\wl \leq \wait,\kappa_\wl \leq \wait, \zeta_\wl \geq \zeta$. Consider some action $b \in \cA$ such that $\Lambda$ payment exploration batches of length $\Texp$ are run on arm $b$ - $\Lambda \in \iint{1}{\lceil \log T^\beta \rceil}$ - for any player $\wl \in \ch(\vl)$. Then
\begin{align*}
    \bigcap_{\lambda \in [\Lambda]} \evgoodw(\iint{k_{b,\lambda}^\vl+1}{k_{b,\lambda}^\vl+\Texp}, \acst, \kappa) \subseteq \{\icvs_b(\wl) \in [\licv_b(\Lambda), \hicv_b(\Lambda)]\} \eqsp,
\end{align*}
and the probability of these events is at least $1- \lceil \log T^\beta \rceil/T^{\alpha \zeta}$. We also have that for any $\wl \in \ch(\vl)$
\begin{align*}
    \{ |\hicv_b(\Lambda) - \licv_b(\Lambda) | \leq 1/2^{\Lambda} +2/T^\beta \} \text{ holds almost surely} \eqsp.
\end{align*}
\end{restatable}

\begin{proof}[Proof of \Cref{lemma:precision_incentives}]
Suppose that the conditions of the lemma hold and consider a node $\wl \in \ch(\vl)$ as well as some arm $b \in \cA$.

The proof is done by induction on the number of binary search batches $\Lambda$ that are run on $b$ and player $\wl$, such that $\bigcap_{\lambda \in [\Lambda]} \evgoodw(\iint{k_{b,\lambda}^\vl+1}{k_{b,\lambda}^\vl+\Texp}, \acst, \kappa) \subseteq \{\icvs_b(\wl) \in [\licv_b(\Lambda), \hicv_b(\Lambda)] \}$.

The initialisation holds since $\licv_a(0) = 0$, $\hicv_a(0) = 1$ and $\icvs_b(\wl) = \max_{b' \in \cA}\reww_{b'}-\reww_b \in [0,1]$ with probability $1$ - since $\max_{b' \in \cA}\reww_{b'} \in [0,1]$ and $\reww_b \in [0,1]$.

Suppose that the property is true for some integer $\Lambda < \lceil \log T^\beta \rceil$ and that we have run one more binary search on arm $b$. We have
\begin{equation*}
\icvmid_b(\Lambda+1) = \frac{\hicv_b(\Lambda) + \licv_b(\Lambda)}{2} \eqsp,
\end{equation*}
$\icvmid_b(\Lambda+1)$ being the transfer offered to $\wl$ if he chooses action $b$ during the $\Lambda$-th batch $\iint{k_{b,\lambda}^\vl+1}{k_{b,\lambda}^\vl+\Texp}$, for the contract $(b,\icvmid_b(\Lambda+1))$. After this batch, if $\Tnot_{b,\Lambda} <\Texp- \acst \Texp^{\kappa+\beta/\alpha}$, \Cref{algorithm:search_subroutine} updates $\hicv_b(\Lambda+1) = \icvmid_b(\Lambda+1) +1/T^\beta$, $\licv_b(\Lambda+1) = \licv_b(\Lambda)$ and \Cref{lemma:classification_search} ensures that $\licv_b(\Lambda+1) < \icvs_b< \hicv_b(\Lambda+1)$ given $\evgoodw(\iint{k_{b,\lambda}^\vl+1}{k_{b,\lambda}^\vl+1\Texp}, \acst, \kappa)$. Thus the induction holds.

Otherwise, if $\Tnot_{a,\Lambda} > \acst \Texp^{\kappa+\beta/\alpha}$, \Cref{algorithm:search_subroutine} updates $\licv_b(\Lambda+1) = \icvmid_b(\Lambda+1)-1/T^\beta$, $\hicv_b(\Lambda+1) = \hicv_b(\Lambda)$ and \Cref{lemma:classification_search} ensures that $\licv_b(\Lambda+1) < \icvs_b< \hicv_b(\Lambda+1)$ given $\evgoodw(\iint{k_{b,\lambda}^\vl+1}{k_{b,\lambda}^\vl+1\Texp}, \acst, \kappa)$. The induction still holds.

Therefore, for any number $\Lambda$ of payment exploration batches run on arm $b$, we have that $\icvs_b(\wl) \in [\licv_b(\Lambda), \hicv_b(\Lambda)]$ with probability at least $1-\Lambda/T^{\alpha \zeta}$. \Cref{lemma:intersection_bounded_regret} guarantees that $\bigcap_{\lambda \in [\Lambda]} \evgoodw(\iint{k_{b,\lambda}^\vl+1}{k_{b,\lambda}^\vl+\Texp}, \acst, \kappa)$ holds with probability at least $1-\Lambda/T^{\alpha \zeta}$.

For the second part of the proof, still considering the child $\wl \in \ch(\vl)$, we define the sequence $u_\Lambda^\wl= \hicv_b(\Lambda) - \licv_b(\Lambda)\geq 1$ as the length of the interval containing $\icvs_b(\wl)$ with probability at least $1-\Lambda/T^{\alpha \zeta}$. We have $u_0^\wl = 1$. Suppose that after $\Lambda$ iterations of binary search batches, the next batch of binary search $\iint{k_{b,\Lambda+1}+1}{k_{b,\Lambda+1}+\Texp}$ outputs $\Tnot_{b,\Lambda+1} < \Texp-\acst \Texp^{\kappa+\beta/\alpha}$. Then, the update of \Cref{algorithm:search_subroutine} gives
\begin{align*}
    u_{\Lambda+1}^\wl 
    = \icvmid_b(\Lambda+1)+1/T^\beta - \licv_b(D) = \frac{\hicv_b(\Lambda)+\licv_b(\Lambda)}{2}-\licv_b(\Lambda) +1/T^\beta  
    =u_{\Lambda}^\wl/2+1/T^\beta \eqsp.
\end{align*}
On the other hand, if $\Tnot_{b,\Lambda+1} > \acst \Texp^{\kappa+\beta/\alpha}$, the update gives
\begin{align*}
    u_{\Lambda+1}^\wl = \hicv_b(\Lambda) - (\icvmid_b(\Lambda+1)-1/T^\beta ) = \hicv_b(\Lambda) - \frac{\hicv_b(\Lambda)+\licv_b(\Lambda)}{2} +1/T^\beta =u_{\Lambda}^\wl/2+1/T^\beta \eqsp.
\end{align*}

Thus, $(u_{\Lambda}^\wl)_{\Lambda \geq 1}$ is an arithmetico-geometric sequence defined by $u_{\Lambda+1}^\wl = u_{\Lambda}^\wl/2 +1/T^\beta$ with an initial term $u_0^\wl=1$. Writing $r= 1/T^\beta/(1-1/2) = 2/T^\beta$, we obtain that
\begin{align*}
    |\hicv_b(\Lambda) - \licv_b(\Lambda)| = u_{\Lambda}^\wl = 1/2^{\Lambda}(1-r)+r= 1/2^{\Lambda}(1-2/T^\beta)+2/T^\beta \leq 1/2^{\Lambda}+2/T^\beta \eqsp,
\end{align*}
for any $\Lambda \in \iint{1}{\lceil \log T^\beta \rceil}$, hence the result.
\end{proof}

\begin{restatable}{proposition}{estimatedincentivesgood}\label{proposition:estimated_incentives_good}
Consider a player $\vl \in \trV$ such that any children $\wl \in \ch(\vl)$ satisfies \Cref{assumption:bounded_regret} with parameters $(\wait_\wl, \acst_\wl, \kappa_\wl, \zeta_\wl)$ and $\wait_\wl \leq \wait, \acst_\wl \leq \acst, \kappa_\wl \leq \kappa, \zeta_\wl \geq \zeta$. Suppose that \Cref{assumption:bounded_rewards} holds and that $\vl$ runs \Cref{algorithm:search_subroutine} with parameters $\wait, \acst, \kappa, \alpha, \beta, \extra$ on any action $b \in \ch(\vl)$ for any player $\wl \in \ch(\vl)$. In that case, we have that the event
\begin{align*}
    & \text{for any }\wl \in \ch(\vl), b \in \cA, \\
    & \icvest_b(\wl)-4/T^\beta -\acst \, \mB\, T^{-\extra} \leq \icvs_b(\wl) \leq \icvest_b(\wl) \eqsp,
\end{align*}
holds with probability at least $1-K \mB \lceil \log T^\beta \rceil /T^{\alpha \zeta}$.
\end{restatable}

\begin{proof}[Proof of \Cref{proposition:estimated_incentives_good}]
We start the proof by considering the opposite event and applying an union bound, following the definition
$\icvest_b(\wl) = \hicv_b(\lceil \log T^\beta \rceil)+1/T^\beta + \acst \, \mB \, T^{-\extra} $
\begin{align*}
& \P\left(\text{for any }\wl \in \ch(\vl), b \in \cA, \icvest_b(\wl)-4/T^\beta -\acst \, \mB \, T^{-\extra} \leq \icvs_b(\wl) \leq \icvest_b(\wl) \right) \\
& \quad \geq 1- \sum_{\wl \in \ch(\vl)} \sum_{b \in \cA} \P\left(\{\icvest_b(\wl)-4/T^\beta -\acst \, \mB \, T^{-\extra} \leq \icvs_b(\wl) \leq \icvest_b(\wl) \}^\mrcc \right) \\
& \quad = 1- \sum_{\wl \in \ch(\vl)} \sum_{b \in \cA} \P\left(\{\hicv_b-3/T^\beta\leq \icvs_b(\wl)\leq \hicv_b +1/T^\beta + \acst \, \mB \, T^{-\extra} \}^\rmcc \right) \eqsp,
\end{align*}
by definition of $\icvest_b(\wl)$ from \eqref{equation:estimated_incentives_definition}. If $\bigcap_{\lambda \in \lceil \log T^\beta \rceil} \evgoodw(\iint{k_{b,\lambda}^\vl+1}{k_{b,\lambda}^\vl+\Texp}, \acst, \kappa)$ holds, then \Cref{lemma:precision_incentives} gives that $\icvs_b(\wl) \in [\licv_b(\lceil \log T^\beta\rceil), \hicv_b(\lceil \log T^\beta\rceil)]$, and the second part of the lemma gives that
\begin{equation*}
|\hicv_b(\lceil \log T^\beta\rceil)-\licv_b(\lceil \log T^\beta\rceil)| \leq 1/2^{\lceil \log T^\beta \rceil} +2/T^\beta \leq 1/T^\beta + 2/T^\beta = 3/T^\beta \eqsp.
\end{equation*}
In that case, we can write
\begin{align*}
& \P\left(\text{for any }\wl \in \ch(\vl), b \in \cA, \icvest_b(\wl)-4/T^\beta -\acst \, \mB \, T^{-\extra}\leq \icvs_b(\wl) \leq \icvest_b(\wl) \right) \\
& \quad \geq 1 - \sum_{\wl \in \ch(\vl)} \sum_{b \in \cA} \P\left(\{\licv_b(\lceil \log T^\beta\rceil) \leq \icvs_b(\wl)\leq \hicv_b(\lceil \log T^\beta\rceil)\}^\rmcc \right) \eqsp.
\end{align*}
However, \Cref{lemma:intersection_bounded_regret} ensures that $\bigcap_{\lambda \in \lceil \log T^\beta \rceil} \evgoodw(\iint{k_{b,\lambda}^\vl+1}{k_{b,\lambda}^\vl+\Texp}, \acst, \kappa)$ happens with probability at least $1- \lceil \log T^\beta \rceil/T^{\alpha \zeta}$. Therefore,
\begin{align*}
& \P\left(\text{for any }\wl \in \ch(\vl), b \in \cA, \icvest_b(\wl)-4/T^\beta -\acst \, \mB \, T^{-\extra} \leq \icvs_b(\wl) \leq \icvest_b(\wl) \right) \\
& \quad \geq 1 - \mB \, K \, \lceil \log T^\beta \rceil /T^{\alpha \zeta} \eqsp.
\end{align*}
\end{proof}

\begin{proof}[Proof of \Cref{lemma:decomposition_regret}]
We first write our regret on the window $\iint{s+1}{s+t}$ depending on whether all the children accepted the contract or not, following
\begin{align*}
    \nonumber \regret^\vl(s,t) & = \sum_{l=s}^{t} \max_{a, b^{\ch(\vl)}} \left\{\mu^\vl(a, b^{\ch(\vl)}) + \1_{B_l^\vl}(a) \, \icv_l(\vl) \right\} \\
    & \quad - \sum_{l=s}^t \left\{ \theta^\pl(A_l^\vl, A_l^{\ch(\vl)}) +\1_{B_l^\vl}(A_l^\vl)\, \icv_l(\vl) - \sum_{\wl \in \up(\pl)} \1_{B_l^\wl}(A^\wl_l)\icv_l(\wl)\right\} \\
    & = \sum_{l=s}^{t} \max_{a, b^{\ch(\vl)}} \left\{\mu^\vl(a, b^{\ch(\vl)}) + \1_{B_l^\vl}(a) \, \icv_l(\vl) \right\} \\
    & \quad - \sum_{l=s}^t \left\{ \theta^\pl(A_l^\vl, B_l^{\ch(\vl)}) +\1_{B_l^\vl}(A_l^\vl)\, \icv_l(\vl) - \sum_{\wl \in \up(\pl)} \icvs_{B_l^\wl}(\wl) \right\} \\
    & \quad + \sum_{l=s}^t \sum_{\wl \in \ch(\vl)} (\1_{B_l^\wl}(A^\wl_l)\icv_l(\wl) -\icvs_{B_l^\wl}(\wl)) + \sum_{l=s}^t (\theta^\vl(A_l^\vl, B_l^{\ch(\vl)}) - \theta^\vl(A_l^\vl, A_l^{\ch(\vl)})) \\
    & \leq  \sum_{l=s}^{t} \left\{ \max_{a, b^{\ch(\vl)}} \{ \mu^\vl(a, b^{\ch(\vl)}) + \1_{B_l^\vl}(a) \, \icv_l(\vl)\} - (\rewv(A_l^\vl, B_l^{\ch(\vl)}) +\1_{B_l^\vl}(A_l^\vl)\, \icv_l(\vl))\right\} \\
    & \quad + \sum_{l=s}^t \sum_{\wl \in \ch(\vl)} (\icv_l(\wl) -\icvs_{B_l^\wl}(\wl))_+ + \sum_{l=s}^t (\theta^\vl(A_l^\vl, B_l^{\ch(\vl)}) - \theta^\vl(A_l^\vl, A_l^{\ch(\vl)}))_+ \\
    & \leq \actionregret^\vl(s,t) +  \paymentregret^\vl(s,t) + \deviationregret^\vl(s,t) \eqsp,
\end{align*}
hence the result.
\end{proof}

\begin{lemma}\label{lemma:bound_regret_decomposition}
Considering the setup of \Cref{lemma:decomposition_regret}, we have that
\begin{align*}
    & \paymentregret^\vl(s,t)\leq \mB \, (t-s+1)\max_{\underset{l \in \iint{s}{t}}{\wl \in \ch(\vl)}} |\icv_{l}(\wl)-\icvs_{B_l^\wl}(\wl)| \eqsp, \\
    & \quad \text{ and } \;\; \deviationregret^\vl(s,t) \leq \Card \left\{l \in \{s,\cdots,t\} \colon A_l^{\ch(\vl)}\ne B_l^{\ch(\vl)}\right\} \eqsp.
\end{align*}
\end{lemma}

\begin{proof}[Proof of \Cref{lemma:bound_regret_decomposition}]
We start with the payment regret, which can be written
\begin{align*}
    \paymentregret^\vl(s,t) & = \sum_{l=s}^t \sum_{\wl \in \ch(\vl)} (\icv_l(\wl) -\icvs_{B_l^\wl}(\wl))_+ \\
    & \leq \sum_{l=s+1}^{s+t} \mB \max_{\wl \in \ch(\vl)}|\icv_{l}(\wl)-\icvs_{B_l^\wl}(\wl)| \\
    & \leq \sum_{l=s+1}^{s+t} \mB \max_{\wl \in \ch(\vl)} \max_{l \in \iint{s+1}{s+t}}|\icv_{l}(\wl)-\icvs_{B_l^\wl}(\wl)| \\
    & \leq \mB \, t \max_{\underset{l \in \iint{s+1}{s+t}}{\wl \in \ch(\vl)}} |\icv_{l}(\wl)-\icvs_{B_l^\wl}(\wl)| \eqsp.
\end{align*}
We can now decompose the deviation regret, following
\begin{align*}
    \deviationregret^\vl(s,t) & = \sum_{l=s}^t (\theta^\vl(A_l^\vl, B_l^{\ch(\vl)}) - \theta^\vl(A_l^\vl, A_l^{\ch(\vl)}))_+ \\
    & = \sum_{l=s}^t \1_{B_l^{\ch(\vl)}}(A_l^{\ch(\vl)}) \underbrace{(\theta^\vl(A_l^\vl, B_l^{\ch(\vl)}) - \theta^\vl(A_l^\vl, A_l^{\ch(\vl)}))_+}_{\leq 1} \\
    & \leq \sum_{l=s}^t \1_{B_l^{\ch(\vl)}}(A_l^{\ch(\vl)}) \\
    & = \Card \left\{l \in \{s,\cdots,t\} \colon A_l^{\ch(\vl)}\ne B_l^{\ch(\vl)}\right\} \eqsp.
\end{align*}
\end{proof}

\begin{proof}[Proof of \Cref{lemma:repeat_assumption}]
    Consider the setup of the lemma, as well as some set of steps $\iint{s+1}{s+t}\subseteq \iint{\Psi_{K+1}^\vl+1}{T}$. We define the following "nice" event
    \begin{align*}
        \evgood = & \bigcap_{\wl \in \ch(\vl)}\bigcap_{\lambda \in \lceil \log T^\beta \rceil} \bigcap_{b \in \cA} \evgoodw(\iint{k_{b,\lambda}^\vl+1}{k_{b,\lambda}^\vl+\Texp}, \acst, \kappa) \\
        & \quad \bigcap_{\wl \in \ch(\vl)} \left\{\sum_{l=s+1}^{s+t} \max_{a \in \cA} \{\reww_a + \1_{B_l^\wl}(a)\, \icv_l(\wl)\} - \left\{ \mu^\wl(A_l^\wl, B_l^{\ch(\wl)}) +\1_{B_l^\wl}(A_l^\wl)\, \icv_l(\wl)\right\} \leq \acst t^\kappa \right\} \eqsp,
    \end{align*}
and a straight computation gives that
\begin{align*}
\P(\evgood) & \geq 1 - \mB \, K \, \lceil \log T^\beta \rceil /T^{\alpha \zeta}-\mB t^{-\zeta} \\
& \geq 1- \mB \, K \, \lceil \log T^\beta \rceil /t^{\alpha \zeta}-\mB/t^{-\alpha \zeta} \\
& \geq 1 - \max\left\{\frac{2 \, \mB\, K \lceil \log T^\beta \rceil}{T^{\alpha \zeta}}, \frac{2 \, \mB}{t^\zeta} \right\} \\
& \geq 1 - \max\left\{\frac{1}{T^{\alpha \zeta-\epsilon}}, \frac{2\, \mB}{t^{\zeta}} \right\} \\
& \geq 1-1/t^{\alpha \zeta -\epsilon} \eqsp,
\end{align*}
for $\epsilon$ defined as
\begin{equation}\label{equation:definition_epsilon}
    \epsilon = \log(4 \mB K \log T)/\log(T) \eqsp,
\end{equation}

and $t \geq (2\mB)^{\zeta - \alpha \zeta+\epsilon}$ with $\zeta - \alpha \zeta+\epsilon>0$. We suppose that $\cE$ holds for the rest of the proof. Based on the children $\wl \in \ch(\vl)$ time to learn, we can define player $\vl$'s time to learn, following $\wait_\vl = \Psi_{K+1}^\vl + \max_{\wl \in \ch(\vl)} \wait_\wl$, and the following sets
\begin{align*}
    I_{\iint{s+1}{s+t}}^\wl &= \{l \in \iint{s+1}{s+t} \text{ such that } A_l^\wl = B_l^\wl\} \eqsp, \\
    J_{\iint{s+1}{s+t}}^\wl &= \{l \in \iint{s+1}{s+t} \text{ such that } A_l^\wl \ne B_l^\wl\} \eqsp,
\end{align*}
as well as
\begin{align*}
    &\Imax_{\iint{s+1}{s+t}}(\vl) = \{l \in \iint{s+1}{s+t} \text{ such that } l \in I_{\iint{s+1}{s+t}}^\wl \text{ for any } \wl \in \ch(\vl)\} \eqsp, \\
    &\Jmax_{\iint{s+1}{s+t}}(\vl) = \{l \in \iint{s+1}{s+t} \text{ such that } l \in J_{\iint{s+1}{s+t}}^\wl \text{ for some } \wl \in \ch(\vl)\} \eqsp.
\end{align*}
Note that $I_{\iint{s+1}{s+t}}\cup J_{\iint{s+1}{s+t}} = \iint{s+1}{s+t}$  and $\Imax_{\iint{s+1}{s+t}} \cup \Jmax_{\iint{s+1}{s+t}} = \iint{s+1}{s+t}$. Since $\cE$ holds, we have that for any $\wl \in \ch(\vl)$, by definition of $\icvest_b(\wl)$
\begin{align*}
    \reww_b + \icvest_b(\wl) > \, \reww_b + \max_{b'' \in \cA} \reww_{b''}- \reww_b + \acst \,\mB\, T^{-\extra} \, \geq \, \reww_{b'} + \acst \,\mB\, T^{-\extra} \eqsp,
\end{align*}
for any $b' \in \cA$ such that $b'\ne b$. Therefore, we have that $B_t^\wl = \argmax_{b \in \cA} \{\reww_b+\1_{B_t^\wl}(b)\, \icvest_b(\wl)\}$, as well as
\begin{equation}\label{equation:reward_gap_incentives}
    (\reww_{B_t^\wl}+\1_{B_t^\wl}(B_t^\wl) \, \icvest_{B_t^\wl}) - \max_{\underset{b\ne B_t^\wl}{b \in \cA}} \{\reww_b+\1_{B_t^\wl}(b) \, \icvest_{B_t^\wl} \} \geq \acst \, \mB\, T^{-\extra} \eqsp.
\end{equation}
Thus, for any step $l \in \iint{s+1}{s+t}$, \eqref{equation:reward_gap_incentives} implies that the reward gap between $B_t^\wl$ and any other action $b \in \cA$ is at least $\acst \,\mB \, T^{-\extra}$, and hence
\begin{align*}
    \Card(J_{\iint{s+1}{s+t}}^\wl) \, \acst \, T^{-\extra}\, \mB & \leq \sum_{l\in \iint{s+1}{s+t}} \max_{a \in \cA} \{\rewv_a + \1_{B_l^\vl}(a)\, \icv_l(\vl)\} \\
    & \quad - \left\{ \mu^\vl(A_l^\vl, B_l^{\ch(\vl)}) +\1_{B_l^\vl}(A_l^\vl)\, \icv_l(\vl)\right\} \\
    & \leq \acst \, t^\kappa \eqsp,
\end{align*}
by \Cref{assumption:bounded_regret}. We can conclude that $\Card(J_{\iint{s+1}{s+t}}^\wl) \leq t^\kappa/(T^{-\extra} \,\cdot \mB) \leq t^{\kappa+\extra}/\mB$ for any $\wl \in \ch(\vl)$, which gives that
\begin{equation}\label{equation:cardjmaxlittle}
    \Card(\Jmax_{\iint{s+1}{s+t}}) \leq \sum_{\wl \in \ch(\vl)} \Card(J_{\iint{s+1}{s+t}}^\wl) \leq  t^{\kappa+\eta} \eqsp.
\end{equation}
We can now decompose the quantity that we want to bound
\begin{align*}
    &\sum_{l=s+1}^{s+t}\left\{ \max_{a \in \cA} \{\rewv_a + \1_{B_l^\vl}(a)\, \icv_l(\vl)\} - \left(\mu^\vl(A_l^\vl, B_l^{\ch(\vl)}) + \1_{B_l^\vl}(A_l^\vl)\, \icv_l(\vl)\right)\right\} \leq 2 \, \Card(\Jmax_{\iint{s+1}{s+t}}) \\
    & \quad + \sum_{l \in \Imax_{\in \iint{s+1}{s+t}}}\left\{ \max_{a \in \cA} \{\rewv_a + \1_{B_l^\vl}(a)\, \icv_l(\vl)\} - \left(\mu^\vl(A_l^\vl, B_l^{\ch(\vl)}) + \1_{B_l^\vl}(A_l^\vl)\, \icv_l(\vl)\right)\right\} \\
    & \leq 2\, t^{\kappa+\extra} + \sum_{l \in \Imax_{\iint{s+1}{s+t}}} \max_{a, b^{\ch(\vl)} \in \cA\times \cA^\mB} \left\{ \theta^\vl(a, b^{\ch(\vl)})-\sum_{\wl \in \ch(\vl)}\icvs_{b^\wl}(\wl) + \1_{B_l^\vl}(a) \icv_l(\vl)\right\} \\
    & \quad - \sum_{l \in \Imax_{\iint{s+1}{s+t}}}\left(\theta^\vl(A_l^\vl,A_l^{\ch(\vl)}) + \1_{B_l^\vl}(A_l^\vl)\, \icv_l(\vl) - \sum_{\wl \in \ch(\vl)}\icvs_{A_l^\wl}(\wl) \right) \\
& = 2\, t^{\kappa+\extra} + \sum_{l \in \Imax_{\iint{s+1}{s+t}}} \max_{a, b^{\ch(\vl)} \in \cA\times \cA^\mB} \left\{ \theta^\vl(a, b^{\ch(\vl)}) \right. \notag \\
& \quad \left. -\sum_{\wl \in \ch(\vl)}(\icvs_{b^\wl}(\wl)+\acst \,\mB\,T^{-\extra}+ 4 T^{-\beta}) + \1_{B_l^\vl}(a) \icv_l(\vl)\right\} \\
    & \quad - \sum_{l \in \Imax_{\iint{s+1}{s+t}}}\left(\theta^\vl(A_l^\vl,A_l^{\ch(\vl)}) + \1_{B_l^\vl}(A_l^\vl)\, \icv_l(\vl) - \sum_{\wl \in \ch(\vl)}(\icvs_{A_l^\wl}(\wl)+\acst\,\mB \, T^{-\extra}) \right) + \, 4 \, \mB\, t \, T^{-\beta} \\
    & \leq 2\, t^{\kappa+\extra} + \, 4 \, \mB\, t^{1-\beta} + \sum_{l \in \Imax_{\iint{s+1}{s+t}}} \max_{a, b^{\ch(\vl)} \in \cA\times \cA^\mB} \left\{ \theta^\vl(a, b^{\ch(\vl)})-\sum_{\wl \in \ch(\vl)}\icvest_{b^\wl}(\wl) + \1_{B_l^\vl}(a) \icv_l(\vl)\right\} \\
    & \quad - \sum_{l \in \Imax_{\iint{s+1}{s+t}}}\left(\theta^\vl(A_l^\vl,A_l^{\ch(\vl)}) + \1_{B_l^\vl}(A_l^\vl)\, \icv_l(\vl) - \sum_{\wl \in \ch(\vl)}\icvest_{A_l^\wl}(\wl) \right) 
    \eqsp,
\end{align*}
thanks to \Cref{proposition:estimated_incentives_good}. But player $\vl$ plays on a stochastic bandit with reward and observed feedback $Z_l(A_l^\vl,A_l^{\ch(\vl)})+\1_{B_l^\vl}(A_l^{\ch(\vl)}) = X_l(A_l^\vl,A_l^{\ch(\vl)}) - \sum_{\wl \in \ch(\vl)}\icvest_{A_t^\wl}(\wl)+\1_{B_l^\vl}(A_l^{\ch(\vl)}$ as long as $l \in \Imax_{\iint{s+1}{s+t}}$. Hence, \citet[Proposition 2]{scheid2024learning} gives us that
\begin{align*}
    \sum_{l=s+1}^{s+t}\left\{ \max_{a \in \cA} \{\rewv_a + \1_{B_l^\vl}(a)\, \icv_l(\vl)\} - \left(\mu^\vl(A_l^\vl, B_l^{\ch(\vl)}) - \1_{B_l^\vl}(A_l^\vl)\, \icv_l(\vl)\right)\right\}  \leq 2(t^{\kappa+\extra} + t^{1-\beta}) & \\
    \quad + 8 \sqrt{K^{\mB+1}\log(K^{\mB+1}T^3)} t^{1/2} \eqsp ,& 
\end{align*}
since player $\vl$ plays on a $K^{\mB+1}$ bandit instance. Hence, $\vl$ satisfies \Cref{assumption:bounded_regret} with parameters $(\wait+K\lceil T^\alpha \rceil \lceil \log T^\beta \rceil, 10 \sqrt{K^{\mB+1}\log(K^{\mB+1}T^3)}, \max\{\kappa+\extra, 1-\beta, 1/2\},\alpha \zeta-\epsilon)$, where $\epsilon = \log(4 \mB K \log T)/\log(T)$.
\end{proof}

\begin{proof}[Proof of \Cref{corollary:boundedregretv}]
    Note that \Cref{lemma:repeat_assumption} directly implies that with probability at least $1-1/t^{\zeta_\vl}$ where $\zeta_\vl>0$ coming from \Cref{lemma:repeat_assumption} and $t=T-\wait - K\lceil T^\alpha \rceil \lceil \log T^\beta \rceil$
    \begin{align*}
        \regret^\vl(T) & \leq \wait+K\lceil T^\alpha \rceil \lceil \log T^\beta \rceil + \actionregret^\vl(\wait+K\lceil T^\alpha \rceil \lceil \log T^\beta \rceil, T) \\
        & \quad + \deviationregret(\wait+K\lceil T^\alpha \rceil \lceil \log T^\beta \rceil, T) + \paymentregret(\wait+K\lceil T^\alpha \rceil \lceil \log T^\beta \rceil, T) \eqsp,
    \end{align*}
    and note that \Cref{assumption:bounded_regret} directly implies that $\actionregret^\vl(\wait+K\lceil T^\alpha \rceil \lceil \log T^\beta \rceil, T)=o(T)$. But at the same time, the \Cref{assumption:bounded_regret} on the children gives that $\deviationregret(\wait+K\lceil T^\alpha \rceil \lceil \log T^\beta \rceil, T) = o(T)$ and the precise estimation of the payments guarantees that $\paymentregret(\wait+K\lceil T^\alpha \rceil \lceil \log T^\beta \rceil, T) = o(T)$. Since $\wait+K\lceil T^\alpha \rceil \lceil \log T^\beta \rceil = o(T)$ and $\zeta_\vl>0$ is given by \Cref{lemma:repeat_assumption}, taking the expectation gives that
    \begin{align*}
        \E[\regret^\vl(T)] = o(T) \eqsp.
    \end{align*}
\end{proof}

\begin{proof}[Proof of \Cref{lemma:everyone_assumption}]

    We do the proof by induction, starting with the layer $d=2$: the layer of players interacting with the leaves. \citet[Proposition 2]{scheid2024learning} gives us that \Cref{algorithm:ucb} for the leaves without offering any transfer satisfies \Cref{assumption:bounded_regret} with parameters $(0,8\sqrt{K\log(KT^3)}, 1/2, 2)$. Thus, following \eqref{equation:definition_good_parameters}, we choose parameters $\extra_2=1/4, \alpha_2=(D+1)/2D, \beta_2=1/4$ for any $\vl$ in layer $d=2$. $\kappa_1+\extra_2=3/4, 1-\beta_2=3/4$. Therefore, \Cref{lemma:repeat_assumption} ensures that $\vl$ satisfies \Cref{assumption:bounded_regret} with parameters $(K\lceil T^\alpha\rceil \lceil \log T^\beta \rceil, 10 \sqrt{K^{\mB+1}\log(K^{\mB+1}T^3}), 3/4, 4/3-\epsilon) = ((d-1) \, K \lceil T^\alpha\rceil \lceil \log T^\beta \rceil,10 \sqrt{K^{\mB+1}\log(K^{\mB+1}T^3}), 1-1/2d, \zeta_d)$ for $d=2$. Thus, our initialization holds.
    
    Suppose that the property is true for some node $\wl$ located at depth $d$ and consider a node $\vl = \pa(\wl)$ at depth $d+1$. $\vl$ runs $\alg$ with parameters from \eqref{equation:definition_good_parameters}. It is important to check if \eqref{equation:condition_parameters} is satisfied. Note that for any $d \in \iint{1}{D}$, we have that
    \begin{align*}
        \frac{\beta_{d+1}}{\alpha_{d+1}} = \frac{1}{2(d+1)} \times \frac{D(d+1)}{(D+1)d} = \frac{D}{D+1}\times \frac{1}{2d} \eqsp,
    \end{align*}
    while we have that
    \begin{equation*}
        1-\kappa_d = 1-1+ \frac{1}{2d} = \frac{1}{2d}> \beta_{d+1}/\alpha_{d+1} \eqsp,
    \end{equation*}
    hence $\alpha_{d+1},\beta_{d+1}$ are acceptable parameters for a player at depth $d$ since they satisfy $\beta_{d+1}/\alpha_{d+1}<1-\kappa_d$. We now have that
    \begin{align*}
        \kappa_{d+1} & = \max(1/2, \kappa_d + \extra_{d+1}, 1-\beta_{d+1})\\
        & = \max\left\{1/2, 1-1/(2\, \cdot d)+1/\{2\,d \,(d+1)\}, 1-1/2(d+1)\right\} \\
        & = \max\left\{1-\frac{1}{2d} \left(1-\frac{1}{d+1}\right), 1-\frac{1}{2(d+1)} \right\}\;\;\; \text{for any d}\geq 3 \\
        & = \max\left\{1-\frac{1}{2d} \frac{d+1-1}{d+1}, 1-\frac{1}{2(d+1)} \right\} \\
        & = 1-\frac{1}{2(d+1)} \eqsp.
    \end{align*}
    Now observe that following \Cref{lemma:everyone_assumption} with $\epsilon$ defined in \eqref{equation:definition_epsilon} as
\begin{equation*}
    \epsilon = \log(4\mB K \log T)/\log(T) \eqsp,
\end{equation*}
we have that $\zeta_{d+1} = \alpha_{d+1}\zeta_d-\epsilon$ and \Cref{lemma:bound_zeta} ensures that $\zeta_d \geq 2/d - d \epsilon$. Since we assumed that $\log T\geq D^2 \log (4 \mB K \log T)$, we have for any $d \in \iint{1}{D}, \log T\geq d^2 \log (4 \mB K \log T)$, which implies $1/d \geq d \log(4 \mB K\log T)/\log T = d\epsilon$, and hence we obtain that $\zeta_d\geq 1/d$.
\end{proof}

\begin{lemma}\label{lemma:bound_zeta}
    If $(\zeta_d)_{d\geq 1}$ is defined as $\zeta_1=2$ and for any $d\geq 1$, $\zeta_{d+1}=\alpha_{d+1} \zeta_d - \epsilon$, for $\epsilon \in (0,1)$, then we have that
    \begin{equation*}
        \zeta_d\geq \frac{2}{d} - \epsilon \left(\frac{d+1}{2}-\frac{1}{d}\right) \geq \frac{2}{d}- d \, \epsilon \eqsp.
    \end{equation*}
\end{lemma}

\begin{proof}[Proof of \Cref{lemma:bound_zeta}]
    First define the sequence $(v_d)_{d\geq 1}$ as $v_1=2$ and then $v_{d+1}=d v_d/(d+1) -\epsilon$. We show by induction that for any $d\geq 1, \zeta_d \geq v_d$. The definition of $\zeta_1=2=v_d$ and now assuming that $\zeta_d\geq v_d$ for some $d\geq 1$, we obtain
    \begin{align*}
        \zeta_{d+1}& =\frac{(D+1)(d-1)}{D \times d}\zeta_d-\epsilon \\
        & \geq \frac{d}{d+1} \zeta_d-\epsilon \\
        & \geq \frac{d}{d+1} v_d-\epsilon \;\;\; \text{by our induction assumption} \\
        &= v_{d+1} \eqsp,
    \end{align*}
    by definition. Some algebra gives that $v_d = \frac{2}{d} + \epsilon \left(\frac{1}{d} -\frac{d+1}{2}\right) \geq \frac{2}{d} - d \, \epsilon$, hence the result.
\end{proof}

\begin{proof}[Proof of \Cref{theorem:main_convergence}]
    Consider player $\vl \in \trV$ located at depth $d$. Since all the nodes run $\alg$, \Cref{lemma:repeat_assumption} gives us that $\ch(\vl)$ satisfy \Cref{assumption:bounded_regret} with parameters $(\wait_\wl, \acst_\wl, \kappa_\wl, \zeta_\wl)$ and we have that $\max_{\wl \in \ch(\vl)} \wait_{\wl} \leq K \lceil T^\alpha \rceil \lceil \log T^\beta\rceil$, $\acst_{\wl} \leq 10 \sqrt{K^{\mB+1} \log(K^{\mB+1}T^3)}$. Also, $\kappa_d=1-1/(2d-2)$ and $\zeta_d \geq 1/d$. Player $\vl$ 's instantaneous regret is always bounded by $2$ almost surely, hence we can write
    \begin{align*}
        \regret^\vl(T) & \leq 2 \, d \, K \lceil T^\alpha \rceil \lceil \log T^\beta\rceil + \sum_{s=\wait + K\lceil T^\alpha \rceil \lceil \log(T^\beta)}^T \max_{a, (b^{\ch(\vl)})} \left\{\mu^\vl(a, b^{\ch(\vl)}) + \1_{B_t^\vl}(a) \, \icv_t(\vl) \right\} \\
    & \quad - \sum_{s=\wait + K\lceil T^\alpha \rceil \lceil \log(T^\beta)}^T  \left\{ \theta^\pl(A_t^\vl, A_t^{\ch(\vl)}) +\1_{B_t^\vl}(A_t^\vl)\, \icv_t(\vl) - \sum_{\wl \in \up(\pl)} \1_{B_t^\wl}(A^\wl_t)\icv_t(\wl)\right\}
    \end{align*}
    Using \Cref{lemma:decomposition_regret}, we can decompose $\regret^\vl$ as follows
    \begin{align*}
        \regret^\vl(T) & \leq 2 \, d \, K \lceil T^\alpha \rceil \lceil \log T^\beta\rceil + \actionregret^\vl((d+1)K \lceil T^\alpha \rceil \lceil \log T^\beta\rceil,T) \\
        & \quad + 2 \, \Card (s \in \iint{1+d K \lceil T^\alpha \rceil \lceil \log T^\beta\rceil}{T} \colon A_s^{\ch(\vl)} \ne B_s^{\ch(\vl)}) \\
        & \quad + \mB \, (T-d K \lceil T^\alpha \rceil \lceil \log T^\beta\rceil) \max_{\underset{s \geq (d+1)K \lceil T^\alpha \rceil \lceil \log T^\beta\rceil}{\wl \in \ch(\vl)}} |\icvest_{B_s^\wl}(\wl)-\icvs_{B_s^\wl}(\wl) | \\
        & \leq 2 \, d \, K \lceil T^\alpha \rceil \lceil \log T^\beta\rceil + \acst_\vl (T-K \lceil T^\alpha \rceil \lceil \log T^\beta\rceil)^{\kappa_\vl} + \mB \, T \, \max_{b\in \cA, \wl \in \trV}|\icvest_b(\wl)-\icvs_b(\wl)| \\
        & \quad + 2 \, \Card (s \in \iint{1+d K \lceil T^\alpha \rceil \lceil \log T^\beta\rceil}{T} \colon A_s^{\ch(\vl)} \ne B_s^{\ch(\vl)}) \eqsp,
    \end{align*}
    with probability at least $1-1/T^{-\zeta_d}$ by \Cref{lemma:repeat_assumption}. Following what is shown in \eqref{equation:cardjmaxlittle}, we have that for any $s \geq \wait$
    \begin{align*}
    \Card(\Jmax_{\iint{s+1}{s+t}}) \leq  t^{\kappa_{d-1}+\eta_d} \eqsp,
\end{align*}
and \Cref{proposition:estimated_incentives_good} gives that
\begin{align*}
    \max_{\wl \in \ch(\vl), b \in \cA} |\icvest_b(\wl)-\icvs_b(\wl)| \leq 4/T^{\beta_d} + \acst \, \mB \, T^{-\extra_d} \eqsp.
\end{align*}
Therefore, we can write, with probability at least $1-1/T^{\zeta_d}$
\begin{align*}
    \regret^\vl(T) & \leq 2 \, d \, K \lceil T^{\alpha_d} \rceil \lceil \log T^{\beta_d} \rceil+ \acst \, T^{\kappa_d} + 2 (T-d \, K \lceil T^{\alpha_d} \rceil \lceil \log T^{\beta_d}\rceil)^{\kappa_{d-1}+\extra_d} \\
    & \quad + \mB\, T \, (4\, T^{-\beta_d} + \acst \, \mB \, T^{-\extra_d}) \\
    & = \bigO(T^{\alpha_d}+T^{\kappa_{d-1}+\extra_d}+T^{1-\beta_d}+T^{1-\extra_d}) \\
    & = \bigO(T^{1-\frac{1}{2d^2}}) \eqsp,
\end{align*}
by definition of the parameters in \eqref{equation:definition_good_parameters} and \Cref{lemma:repeat_assumption}. For the comment that follows \Cref{theorem:main_convergence}, we can write (since $\zeta_d\geq 1/d$)
\begin{align*}
    \E\parentheseDeux{\regret^\vl(T)}  \leq (1-1/T^{\zeta_d}) \bigO(T^{1-\frac{1}{2d^2}}) + 2 \times T /T^{\zeta_d} 
    \leq \bigO(T^{1-\frac{1}{2d^2}}) + \bigO(T^{1-\frac{1}{d}})
    = \bigO(T^{1-\frac{1}{2d^2}}) \eqsp,
\end{align*}
hence the fact that $\E\parentheseDeux{\regret^\vl(T)} = o(T)$.
\end{proof}

\begin{proof}[Proof of \Cref{corollary:final_convergence}]
Following \Cref{lemma:optimal_incentives}, we can write
\begin{align*}
    \sum_{t =1}^T \max_{(a^\vl) \in \cA^{|\trV|}} \left\{\sum_{\vl \in \trV}  \theta^\vl(a^\vl, a^{\ch(\vl)})- \theta^{\vl}(A_t^\vl, A_t^{\ch(\vl})\right\} = \sum_{t =1}^T \sum_{\vl \in \trV} \max_{a^\vl \in \cA}\rewv(a^\vl) - \sum_{t =1}^T \sum_{\vl \in \trV} \theta^\vl(A_t^\vl,A_t^{\ch(\vl)}) \eqsp.
\end{align*}
Since all the players run $\alg$, writing $\wait = (D+1) \, K \,\lceil T^\alpha \rceil \lceil \log T^\beta \rceil \geq \sum_{d=1}^D \max_{\wl \in \trV_d} \wait_\wl$, we can write
\begin{align*}
    &\sum_{t =1}^T \max_{(a^\vl) \in \cA^{|\trV|}} \left\{\sum_{\vl \in \trV}  \theta^\vl(a^\vl, a^{\ch(\vl)})- \theta^{\vl}(A_t^\vl, A_t^{\ch(\vl})\right\} \\
    & \quad \leq 2 \, \wait + \sum_{t = \wait +1}^T \sum_{\vl \in \trV} \max_{a^\vl \in \cA}\rewv(a^\vl) - \sum_{t = \wait+ 1}^T \sum_{\vl \in \trV} (\theta^\vl(A_t^\vl,A_t^{\ch(\vl)})+\1_{B_t^\vl}(A_t^\vl)\icv_t(\vl) -\1_{B_t^\vl}(A_t^\vl)\icv_t(\vl)) ,
\end{align*}
and reorganizing the terms $\1_{B_t^\vl}(A_t^\vl)\icv_t(\vl)$ among the players allows us to write
\begin{align*}
    &\sum_{t =1}^T \max_{(a^\vl) \in \cA^{|\trV|}} \left\{\sum_{\vl \in \trV}  \theta^\vl(a^\vl, a^{\ch(\vl)})- \theta^{\vl}(A_t^\vl, A_t^{\ch(\vl})\right\} \\
    & \quad \leq 2 \, \wait + \sum_{t = \wait +1}^T \sum_{\vl \in \trV} \max_{a^\vl \in \cA}\rewv(a^\vl) \\
    & \quad - \sum_{t=\wait+1}^T \sum_{\vl \in \trV} \left(\theta^\vl(A_t^\vl,A_t^{\ch(\vl)})+\1_{B_t^\vl}(A_t^\vl)\icv_t(\vl) -\sum_{\wl \in \ch(\vl)}\1_{B_t^\wl}(A_t^\vl)\icv_t(\wl) \right) \\
    & \quad \leq 2 \, \wait + \sum_{t = \wait +1}^T \sum_{\vl \in \trV} \max_{a^\vl \in \cA}\rewv(a^\vl) \\
    & \quad -\sum_{t=\wait+1}^T\sum_{\vl \in \trV} \left(\theta^\vl(A_t^\vl,A_t^{\ch(\vl)})-\sum_{\wl \in \ch(\vl)} \icvest_{B_t^\wl}(\wl) + \1_{B_t^\vl}(A_t^\vl)\icv_t(\wl)  \right) \\
    & \quad \leq 2 \wait + \sum_{t = \wait +1}^T \sum_{\vl \in \trV} \max_{a^\vl \in \cA}\rewv(a^\vl) \\
    & \quad -\sum_{t=\wait+1}^T\sum_{\vl \in \trV} \left(\theta^\vl(A_t^\vl,B_t^{\ch(\vl)})-\sum_{\wl \in \ch(\vl)} \icvs_{B_t^\wl}(\wl) + \1_{B_t^\vl}(A_t^\vl)\icv_t(\vl)  \right) \\
    & \quad \quad \sum_{\vl \in \trV}\left\{ T\, \max_{b \in \cA, \wl \in \ch(\vl)}|\icvs_b(\wl)-\icvest_b(\wl)| + \Card\left(t \geq \wait+1 \colon A_t^{\ch(\vl)}\ne B_t^{\ch(\vl)} \right) \right\} \\
    & \quad \leq 2 \wait + \sum_{t = \wait +1}^T \sum_{\vl \in \trV} \max_{a^\vl \in \cA}\rewv(a^\vl) -\sum_{t=\wait+1}^T\sum_{\vl \in \trV} \left(\mu^\vl(A_t^\vl,B_t^{\ch(\vl)}) +\1_{B_t^\vl}(A_t^\vl)\icv_t(\vl)\right) \\
    & \quad \quad + \sum_{\vl \in \trV}\left\{T \, (4/T^{\min_d \beta_d} + \mB \, \max_{d} \acst_d \, T^{-\min_d \extra_d}) + T^{\max_d \kappa_d+\extra_d} \right\} \\
    & \quad \leq 2\, (D+1) \, K \,\lceil T^\alpha \rceil \lceil \log T^\beta \rceil \\
    & \quad + \sum_{\vl \in \trV}\left(\regret^\vl (T) + 4 \, T^{1-\min_d \beta_d}+ \mB \max_d \acst_d \, T^{1-\min_d \extra_d} + T^{max_d \kappa_d+\extra_d}\right) \eqsp,
\end{align*}
with probability at least $1-\sum_{\vl \in \trV}1/T^{\zeta_{\vl}} \to 0$ as $T \to +\infty$. Therefore, dividing by $T$, taking the expectation as in the proof of \Cref{theorem:main_convergence}, and using \Cref{theorem:main_convergence} for the convergence of any $\regret^\vl(T)$, we obtain that
\begin{align*}
    \sum_{t =1}^T \max_{(a^\vl) \in \cA^{|\trV|}} \left\{\sum_{\vl \in \trV}  \theta^\vl(a^\vl, a^{\ch(\vl)})- \theta^{\vl}(A_t^\vl, A_t^{\ch(\vl})\right\}/T \to 0 \eqsp,
\end{align*}
and therefore
\begin{equation*}
    \sum_{t =1}^T \max_{(a^\vl) \in \cA^{|\trV|}} \left\{\sum_{\vl \in \trV}  \theta^\vl(a^\vl, a^{\ch(\vl)})- \theta^{\vl}(A_t^\vl, A_t^{\ch(\vl})\right\} = o(T) \eqsp,
\end{equation*}
hence our result which gives the convergence of the players' global reward towards the optimum.
\end{proof}

\begin{proof}[Proof of \Cref{lemma:lowerboundish}]
    Consider a two-arms bandit instance $\{0,1\}$ for both the parent ($\vl$) and the child ($\wl$) such that the agent's best action is $0$ while the principal's best action is $(1,1)$. We write $\theta^\vl(i,j)$ the principal's mean reward when the agent plays action $j$ and the principal action $i$. Assume that $\theta^\vl(1,j) > \theta^\vl(0,j)$ for any $j\in \{0,1\}$ and that the players interact for $T$ rounds. We look at the ideal case where the principal knows the minimal payment $\tau^\star>0$ that is needed to make action $1$ the agent's best choice and where $\theta^\vl(1,1) - \tau^\star > \theta^\vl(1,0)+c$, $c$ being some constant big enough, chosen to force the principal to get the agent playing action $1$. At each step $t \in \{1,\ldots ,T\}$, the parent's optimal strategy is to play action $1$ and to offer a transfer $\tau^\star+p_t, p_t\geq 0$ to the child (on this instance, offering lower than $\tau^\star$ would lead the agent to play action $0$) so that the child plays action $1$ (because of the strict dominance of action $1$ for the parent). Writing $\cI_T$ for the set of steps at which the child accepted the incentive (i.e. he best-responded) and $\cJ_T$ for the others, the child's regret is
    \begin{align*}
        \regret^\wl(T) & = \sum_{t\in \cJ_T} p_t= T^\kappa \eqsp,
    \end{align*}
    since $p_t$ represents by definition the agent's reward gap at each step. We match the child's regret upper bound and dismiss the constants. For any sequence $p_t$, considering a reward gap of $1$, the parent's regret is then
        \begin{align*}
        \regret^\vl(T) &= \Card(\cJ_T) + \sum_{t\in \cI_T} p_t \eqsp.
    \end{align*}
    since the principal's regret is $p_t$, the extra-payment for the rounds ($t \in \cI_T$) at which the child accepted the payment and $1$ (or similar kind of constant) whenever the agent refuses the incentive: it corresponds to the rounds in $\cI_T$. Now, for any sequence $\{p_t\}_{t=1}^T$ with $p_t\in [0,1]$, define
$$
S = \left\{t\in\{1,\dots,T\}: p_t \le \tau\right\},\quad \text{with}\quad \tau = T^{-\frac{1-\kappa}{2}},
$$
and set its complement as $S^c = \{1,2,\dots,T\} \setminus S$. We consider two cases.

\vspace{2mm}
\noindent\textbf{Case 1.} $\Card(S) \ge T^{\frac{\kappa+1}{2}}$, then choose $\cJ_T \subseteq S \text{ such that } \Card(\cJ_T) = T^{\frac{\kappa+1}{2}}$. For every $t\in S$, $p_t \le \tau$, thus we have
$$
\regret^\wl(T) = \sum_{t\in J_T} p_t \le \Card(\cJ_T) \cdot \,\tau = T^{\frac{\kappa+1}{2}} \cdot T^{-\frac{1-\kappa}{2}} = T^\kappa \eqsp,
$$
hence the regret constraint is satisfied for the child. Moreover, the parent's regret is now
$$
\regret^\vl(T) = \Card(\cJ_T) + \sum_{t\in \cI_T} p_t \geq  \Card(\cJ_T) = T^{\frac{\kappa+1}{2}} \eqsp.
$$

\vspace{2mm}
\noindent\textbf{Case 2.} $\Card(S) < T^{\frac{\kappa+1}{2}}$, choose $J_T = S$.
Then, as in the previous case, the constraint is satisfied for the child
$$
\regret^\wl(T) = \sum_{t\in J_T} p_t \le \Card(S) \cdot \,\tau < T^{\frac{\kappa+1}{2}} \cdot T^{-\frac{1-\kappa}{2}} = T^\kappa.
$$
The parent's regret is now
$$
\regret^\vl(T) =  \Card(\cJ_T) + \sum_{t\in \cI_T} p_t \geq \sum_{t\in \cI_T} p_t \eqsp.
$$
While $\Card(S)$ is smaller than $T^{\frac{\kappa+1}{2}}$, observe that for every $t\in \cI_T = S^c$, we have $p_t > \tau = T^{-\frac{1-\kappa}{2}}$. Thus,
$$
\sum_{t\in \cI_T} p_t \geq \Card(S^c) \cdot \, T^{-\frac{1-\kappa}{2}} \geq (T - T^{\frac{\kappa+1}{2}}) \cdot T^{-\frac{1-\kappa}{2}} = T^{\frac{\kappa+1}{2}} - T^\kappa = \Omega(T^{\frac{\kappa+1}{2}}) \eqsp.
$$

Therefore, it shows that for any sequence of actions and incentives that the parent can choose, the parent's regret is always lower-bounded, following
      \begin{equation*}
        \regret^\vl(T) \geq \Omega(T^{\frac{\kappa+1}{2}}) \eqsp.
    \end{equation*}
\end{proof}

\section{Equilibrium}

For any outcome $x_t^\vl = (A_t^\vl, B_t^{\ch(\vl)}, \icv_t^{\ch(\vl)}) \in [K]^{\mB+1} \times [0,1]^\mB$, we can define an associated utility
\begin{align*}
U(x^\vl, x^{-\vl})& = \theta^\vl(A^\vl, A^{\ch(\vl)})+\1_{B^\vl}(A^\vl) \icv(\vl) - \sum_{\wl \in \ch(\vl)} \1_{B^\wl}(A^\wl) \icv(\wl) \eqsp.
\end{align*}

We show here that the action profile described by \Cref{lemma:optimal_incentives} corresponds exactly to the \textit{Subgame Perfect Nash Equilibrium} of the extensive form game played at each round~\citep{osborne2004introduction}. In a full information game, $\vl$ would observe everything that occured above her parent $\pa(\vl)$. Here, it does not change anything on $\vl$'s utility and for the sake of clarity, we omit such dependencies in the notation. \Cref{lemma:optimal_incentives} makes the incentives explicit and helps us relate \cref{equation:mu_utility_star} to a notion of equilibrium. Once the players in $\trV$ take the set of actions maximizing $\sum_\vl \rewv(a^\vl)$, the cost of deviating for any player is larger than the reward she can expect from a deviated action.  Before stating the definition of a subgame perfect Nash equilibrium (SPNE), we start by the definition of the strategy profile.

Formally, we define the extended class of strategies of any player $\vl$ as $\cF$ where 
\begin{equation}\label{definition:class_answer}
    \begin{aligned}
    \cF \colon & [K]\times [0,1] \to [K]^{\mB+1} \times [0,1]^\mB \\
    & (B^\vl, \tau(\vl)) \mapsto x^\vl = (A^\vl, B^{\ch(\vl)}, \icv^{\ch(\vl)}) \eqsp.
    \end{aligned}
\end{equation}

Following \Cref{equation:definition_policy}, any online algorithm can be seen as a generator of strategies $f_t^\vl\in\cF$ at any step $t\in \iint{1}{T}$ given by
\begin{equation}\label{eq:extended_action}
    f_t^\vl = \pi^\vl(U_{t}^\vl, \cH_{t-1}^\vl, \cdot, \cdot) \eqsp.
\end{equation}

We now define the utility of a player $\vl$. As seen in \eqref{equation:utility_last_layer}, $\vl$'s utility at step $t$ depends on $A_t^\vl$ but also on $B_t^\vl, (\icv_t(\wl))_\wl, A_t^{\ch(\vl)}$. As long as the recommended action and incentives by her parent $B_t^\vl, \icv_t(\vl)$ is fixed, $f_t^\vl(B_t^\vl, \icv_t(\vl)) = x_t^\vl \in [K]^{\mB+1} \times [0,1]^\mB$ becomes fixed too. Hence, omitting $t$, we write
\begin{equation}\label{equation:definition_utility}
\begin{aligned}
    U^\vl(f^\vl, f^{\ch(\vl)}, B^\vl, \icv(\vl)) & = \theta^\vl(A^\vl, A^{\ch(\vl)})+\1_{B^\vl}(A^\vl) \icv(\vl) -\sum_{\wl \in \ch(\vl)} \1_{B^\wl}(A^\wl) \icv(\wl) \eqsp,
\end{aligned}
\end{equation}
and under our notations, the utility can also be written--if we consider actions $A^{\ch(\vl)}_f$ taken under children strategies $f^{\ch(\vl)}$ for $\vl$'s actions being $f^\vl(B^\vl, \icv(\vl))$--as
\begin{equation*}
    U^\vl(f^\vl, f^{\ch(\vl)}, B^\vl, \icv(\vl)) = U^\vl(f^\vl(B^\vl, \icv(\vl)), A^{\ch(\vl)}_{f^\vl}) \eqsp,
\end{equation*}
where we implicitly drop the dependency of $f^\vl(B^\vl, \icv(\vl))$ on $B^\vl, \icv(\vl)$ when writing $A^{\ch(\vl)}_{f^\vl}$.

\begin{definition}[Subgame Perfect Nash Equilibrium]\label{definition:spne} We say that the profile $(f^\vl)_{\vl \in \trV}$ is an $\epsilon$-subgame perfect Nash equilibrium (SPNE) if for any $\vl$, $B^{\vl}$, $\tau(\vl)$ we have that
\begin{equation}\label{equation:definition_spne}
\begin{aligned}
& \sup_{f \in \cF} U^\vl(f(B^\vl, \icv(\vl)), A^{\ch(\vl)}_f) - U^\vl(f^\vl(B^\vl, \icv(\vl)), A^{\ch(\vl)}_{f^\vl}) \leq \epsilon \eqsp.
\end{aligned}
\end{equation}
\end{definition}
When $\epsilon=0$, we talk about a \textit{Subgame Perfect Equilibrium}. Consider a profile of actions defined as 
\begin{equation}\label{equation:optimal_action}
    x^{\star, \vl} = (A^{\star, \vl}, B^{\star, \ch(\vl)}, \icvs_{\ch(\vl)}(B^{\star, \ch(\vl)})) \eqsp,
\end{equation}
with $\mu$ defined in \Cref{lemma:optimal_incentives} and
\begin{equation*}
    A^{\star, \vl} = B^\vl, B^{\star, \ch(\vl)} = \argmax_{b^{\ch(\vl)}} \{ \mu(A^{\star, \vl},b^{\ch(\vl)})\} \eqsp.
\end{equation*}

\begin{assumption}\label{assumption:unique_argmax}
    For any player $\vl \in \trV$, we have that
    \begin{align*}
    \argmax_{a \in \cA} \mu^{\star, \vl}(a) \qquad \text{is a singleton.}
    \end{align*}
\end{assumption}

\begin{lemma}\label{lemma:spne_outcome}
    Assume that \Cref{assumption:unique_argmax} holds. Suppose that the strategies of players in $\trV$ are a subgame perfect Nash equilibrium $(f^\vl)_{\vl \in \trV}$. Then, the associated profiles of actions played under these strategies are $(x^{\star, \vl})_{\vl \in \trV}$ as defined in \Cref{{equation:optimal_action}}.
\end{lemma}

\begin{proof}[Proof of \Cref{lemma:spne_outcome}]

We justify the result with a backward induction structure on the result.

\textbf{Initialization.} Consider a player $\vl$ being a leaf. Since $\vl$ has no child, the utility she gets for any action $a \in \cA$ is
\begin{equation*}
    \mu^\vl(a) = \theta^\vl(a) \; \text{  and  } \; \rewv(a) = \theta^\vl(a) \eqsp.
\end{equation*}
Now we consider player $\ul = \pa(\vl)$. If $\ul$ uses incentives $\icv(\vl) < \icvs_b(\vl)$ for any recommended action $b$, $\vl$ would play $A^\vl = \argmax_{a \in \cA} \mu^\vl(a)$ as see in the proof of \Cref{lemma:optimal_incentives} to maximize his utility--and the same would happen for all of $\ul$'s children. Therefore $\ul$'s utility for a played action $a \in \cA$ would be
\begin{align}\label{equation:ejflzkfg}
    U^\ul(a,b^{\ch(\ul)}, \icv_{b^{\ch(\ul)}}(\ch(\vl)), B^\ul, \icv(\ul)) & = \theta^\ul(a, \{\argmax_{b}\theta^\wl(b)\}_{\wl \in \ch(\ul)}) + \1_{B^{\ul}}(a)
    \\
    & < \max_{a,b^{\ch(\ul)}} \{\mu^{\star, \ul}(a,b^{\ch(\ul)}) + \1_{B^{\ul}}(a) \} \eqsp,
\end{align}
as shown in the proof of \Cref{lemma:optimal_incentives}, under \Cref{assumption:unique_argmax}. Choosing the smallest $\icv_{b^\wl}(\wl)$ to still enforce action $b^\wl$ for player $\wl$ leads to choosing incentive $\icvs_{b^\wl}(\wl)$.
\begin{align*}
    U^\vl(f, B^\vl, \icv(\vl)) & = \theta^\vl(A^\vl, A^{\ch(\vl)})+\1_{B^\vl}(A^\vl) \icv(\vl) \leq \theta(B^\vl) + \icvs_{B^\vl}(\vl) = \max_{a \in \cA} U(a, B^\vl, \icv(\vl)) \eqsp,
\end{align*}
as we showed. By definition of a SPNE, we thus have that
\begin{align*}
    U^\vl(f^\vl, f^{\ch(\vl)}, B^\vl, \icv(\vl)) = \max_{a \in \cA} U(a, B^\vl, \icv(\vl)) \eqsp,
\end{align*}
which proves that the actions played by $\vl$ under a strategy $f^\vl$ being a SPNE are given by \Cref{equation:optimal_action}.

\textbf{Induction step.} Consider a player $\vl$ not being a leaf. Under our induction hypothesis, her children have action described by \Cref{equation:optimal_action} and all the players' strategies form a SPNE. As shown in the proof of \Cref{lemma:optimal_incentives}, the optimal utility $\vl$ can derive for some action $a \in \cA$ is
\begin{align*}
    \mu^{\star, \vl}(a) = \max_{b^{\ch(\vl)}} \mu^\vl(a,b^{\ch(\vl)}) \eqsp,
\end{align*}
and therefore the maximizing utility--for $\ul=\pa(\vl)$ offering action-incentive pair $B^\vl, \icv_{B^\vl}(\vl)$--is given as
\begin{align*}
    \max_{a \in \cA} \{\mu^{\star, \vl}(a) +\1_{B^\vl}(a) \icv(\vl)\} \eqsp.
\end{align*}
The same reasoning as in the initialization and the proof of \Cref{lemma:optimal_incentives} gives that $\ul$ should offer $\icvs_{B^{\vl}}(\vl)$ to maximize her utility, making $\mu^{\star, \vl}(B^\vl) + \icvs_{B^\vl}(\vl)$ the highest utility $\vl$ can achieve. Since our induction hypothesis gives us that $\vl$'s children under their strategy $f^{\ch(\vl)}$ pick recommended action $B^{\ch(\vl)}$ for incentives as least $\icvs_{B^{\ch(\vl)}}(\ch(\vl))$, we get that $\vl$ can achieve utility
\begin{align*}
    \max_{a \in \cA} \mu^{\star, \vl}(a) = \mu^{\star, \vl}(B^\vl) + \1_{B^\vl}(B^\vl) \icvs_{B^\vl}(\vl) = \max_{a,b^{\ch(\vl)}, \icv_{b^{\ch(\vl)}}(\ch(\vl))} \theta(a,b^{\ch(\vl)}) - \sum_{\wl \in \ch(\vl)} \1_{b^{\wl}}(B^\wl) \icv_{b^\wl}(\wl) + \1_{B^\vl}(a) \icvs_{B^\vl}(\vl) \eqsp,
\end{align*}
by choosing actions $A^\vl = B^\vl, B^{\ch(\vl)} = \argmax_{b^{\ch(\vl)}}\mu^\vl(A^\vl, b^{\ch(\vl)}), \icv(\ch(\vl)) = \icvs_{B^{\ch(\vl)}}(\ch(\vl))$. Therefore, under a strategy $f^\vl \in \cF$ maximizing $\vl$'s utility, since it belongs to a SPNE, we necessarily have that the actions played by $\vl$ under this strategy are given by \Cref{equation:optimal_action}, hence our induction step. Thus, we obtain the result.
\end{proof}

We now define the \textit{empirical joint distribution} as
\begin{equation}\label{equation:definition_empirical_frequency}
    \hat{x}^\vl_T \;\coloneqq\; \frac{1}{T}\sum_{t=1}^T \delta_{x^\vl_t} \in \Delta([K]^{\mB+1} \times [0,1]^\mB) \eqsp,
\end{equation}
where $\delta$ stands for the Dirac distribution. $\hat{x}^\vl_T$ is as the empirical distribution of player $\vl$'s actions. Note that for any measurable $\phi:[K]^{\mB+1} \times [0,1]^\mB \to\R$,
$\E_{x \sim \hat{x}^\vl_T}[\phi(x)]=\frac{1}{T}\sum_{t=1}^T \phi(x^\vl_t)$. Our goal is to show that if the players run $\alg$, their empirical distribution of actions converges to $x^{\star, \vl}$.

Let $\mu$ and $\nu$ be two probability measures on a metric space $(X, d)$. The Wasserstein-$1$ distance between them is defined as
\[
W_1(\mu, \nu) = \inf_{\gamma \in \Pi(\mu, \nu)} 
\int_{X \times X} d(x, y) \, \dist \gamma(x, y)
\]
where $\Pi(\mu, \nu)$ is the set of all couplings $\gamma(x, y)$ whose marginals are $\mu$ and $\nu$ and $\dist(x, y)$ is the distance between points $x$ and $y$. In our case, we write the distance $\dist$ between two action $x,y = (a^x,(b_i^x)_{i \in [K]^\mB}, (\icv_i^x)_{i \in [K]^\mB}), (a^y,(b_i^y)_{i \in [K]^\mB}, (\icv_i^y)_{i \in [K]^\mB})$ as
\begin{align*}
    \dist(x,y) = \1(a^x\ne  a^y) + \sum_{i \in [\mB]} (\1(b^x_i\ne b^y_i) + \1(b^x_i=b^y_i) |\icv_i^x-\icv_i^y|) \eqsp.
\end{align*}
Note that within the space $[K]^{\mB+1}, [0,1]^\mB$, symmetry, separation and triangular inequality are trivially satisfied for $\dist$. By the Kantorovich–Rubinstein duality, it can be rewritten
\[
W_1(\mu, \nu) = 
\sup_{\|f\|_{\text{Lip}} \le 1}
\left(
\int f(x)\, d\mu(x) - \int f(x)\, d\nu(x)
\right) \eqsp,
\]
where the supremum is over all $1$-Lipschitz functions $f$. When considering the Wasserstein-1-distance between the Dirac measure $\delta_{x^{\star, \vl}}$ for $\vl$'s optimal action $x^{\star, \vl}$ and the empirical joint distribution $\hat{x}^\vl_T$, there is a unique admissible coupling given by the product of both distributions, following
\begin{equation}
    \gamma(dx, dy)
    \;=\;
    \delta_{x^{\nu,\star}}(dx)\, \hat{x}^\vl_T(dy)
    \;=\;
    \frac{1}{T}\sum_{t=1}^T \delta_{(x^{\star, \vl}, x_t^\vl)}(dx, dy) \eqsp.
\end{equation}
Substituting this expression into the definition of $W_1$ gives:
\begin{equation}\label{eq:w1_dirac_empirical}
    W_1\!\big(\delta_{x^{\star,\vl}}, \hat{x}^\vl_T\big)
    \;=\;
    \int \dist (x, y)\, d\gamma(x, y)
    \;=\;
    \frac{1}{T}\sum_{t=1}^T \dist \left(x^{\star,\vl},\, x_t^\vl\right) \eqsp.
\end{equation}
Therefore, the 
Wasserstein--1 distance reduces to the \emph{average distance} between 
the point $x^{\star,\vl}$ and the support points $\{x_t^\vl\}_{t=1}^T$ 
of the empirical distribution. Intuitively, since all mass of $\delta_{x^{\nu,\star}}$ is concentrated 
at a single point, there is no optimization over couplings: the entire 
mass must be transported from $x^{\nu,\star}$ to each empirical atom 
$x_t^\nu$, resulting in the average transport cost above. Define the utility gaps (or margins) of a player $\vl \in \trV$ as
\begin{align*}
& \Delta_{b,\wl}^\vl
= \min_{b\neq b^{\star,\wl}}
\Big(\mu^\vl(a^{\star,\vl}, b^{\star,\ch(\vl)}) -
\mu^\vl(a^{\star,\vl}, b^{\star,\ch(\vl) \backslash \wl},b)\Big) > 0 \eqsp, \\
& \Tilde{\Delta}_{b}^\vl = \min_{a \in \cA} \min_{\Tilde{b}\neq \Tilde{b}^{\star,a}\in [K]^\mB} (\theta(a,\Tilde{b}^{\star,a})-\theta(a,\Tilde{b})) \qquad \text{for} \qquad \Tilde{b}^{\star,a} = \argmax_{\Tilde{b} \in [K]^\mB} \; \theta(a,\Tilde{b})
\end{align*}
\begin{assumption}\label{assumption:gap_reward}
    For any player $\vl \in \trV$, we have that $\Delta_{b,\wl}^\vl > 0, \Tilde{\Delta}_{b}^\vl >0$.
\end{assumption}

\medskip

\begin{lemma}\label{lemma:spne_outcome_convergence}
    Assume that \Cref{assumption:unique_argmax} and \Cref{assumption:gap_reward} hold. If all the players run $\alg$, then we have that for any $\vl\in\trV$, with probability greater than $1-o(1)$
    \begin{align*}
    W_1\!\big(\delta_{x^{\star,\vl}}, \hat{x}^\vl_T\big) = \frac{1}{T}\sum_{t=1}^T \dist \left(x^{\star,\vl},\, x_t^\vl\right) = o(1) \eqsp,
    \end{align*}
    which means that the joint empirical distribution of actions converges towards the profile of actions played under a set of players' strategies which is a SPNE.
\end{lemma}

\begin{proof}[Proof of \Cref{lemma:spne_outcome_convergence}]
Assume that all the players $\vl \in \trV$ run $\alg$. For each player $\vl$, we have the empirical measure
\[
\hat{x}_T^\vl \;=\; \frac{1}{T} \sum_{t=1}^T \delta_{x_t^\vl},
\qquad \text{with} \qquad
x_t^\vl = (A_t^\vl,\, B_t^{\ch(\vl)},\, \icv_t(\ch(\vl))).
\]
The target action is $x^{\star,\vl} = (a^{\star,\vl},\, b^{\star,\ch(\vl)},\, \icv^\star(\ch(\vl)))$, where
$\icv^\star(\wl) = \icvs_{b^{\star,\wl}}(\wl)$ denotes the optimal incentive to player $\wl$ for action $b^{\star, \wl}$
defined in \Cref{lemma:optimal_incentives}. We want to control
\begin{align*}
  W_1(\delta_{x^{\star,\vl}}, \hat{x}_t^\vl) &= \frac{1}{T} \sum_{t=1}^T\left\{\1\{A_t^\vl \neq a^{\star,\vl}\}
+\sum_{\wl \in\ch(\vl)}
\Big(\1\{B_t^\wl \neq b^{\star,\wl}\}
+\1\{B_t^\wl=b^{\star,\wl}\}\,|\icv_t(\wl)-\icvs(\wl)|\Big)\right\} \eqsp.
\end{align*}
Consider the event $\cE$, defined as
\begin{align*}
    \cE = \bigcap_{\wl \in \ch(\vl), b \in \cA} \, \left\{ \icvest_b(\wl)-4/T^\beta -\acst \, \mB\, T^{-\extra} \leq \icvs_b(\wl) \leq \icvest_b(\wl) \right\}\bigcap_{b \in \cA} \left\{\icvest_b(\vl)-4/T^\beta -\acst \, \mB\, T^{-\extra} \leq \icvs_b(\vl) \leq \icvest_b(\vl) \right\} \eqsp.
\end{align*}
\Cref{proposition:estimated_incentives_good} and an union bound ensure that $\cE$ holds with probability at least $1-2 \,K \mB \lceil \log T^\beta \rceil /T^{\alpha \zeta} = 1-o(1)$. We do the rest of the proof assuming that $\cE$ holds.

As done in the proof of \Cref{lemma:decomposition_regret}, we can write
\begin{align*}
    \regret^\vl(T) & = \sum_{t=1}^{T} \max_{a, b^{\ch(\vl)}} \left\{\mu^\vl(a, b^{\ch(\vl)}) + \1_{B_t^\vl}(a) \, \icv_t(\vl) \right\} \\
    & \quad - \sum_{t=1}^T \left\{ \theta^\pl(A_t^\vl, B_t^{\ch(\vl)}) +\1_{B_t^\vl}(A_t^\vl)\, \icv_t(\vl) - \sum_{\wl \in \up(\pl)} \icvs_{B_t^\wl}(\wl) \right\} \\
    & \quad + \sum_{t=1}^T \sum_{\wl \in \ch(\vl)} (\1_{B_t^\wl}(A^\wl_t)\icv_t(\wl) -\icvs_{B_t^\wl}(\wl)) + \sum_{t=1}^T (\theta^\vl(A_t^\vl, B_t^{\ch(\vl)}) - \theta^\vl(A_t^\vl, A_t^{\ch(\vl)})) \\   
    & = \sum_{t=1}^{T} \max_{a} \left\{\mu^{\star,\vl}(a) + \1_{B_t^\vl}(a) \, \icv_t(\vl) -\mu^\pl(A_t^\vl, B_t^{\ch(\vl)}) - \1_{B_t^\vl}(A_t^\vl)\, \icv_t(\vl)\right\} \\
    & \quad + \sum_{t=1}^T \sum_{\wl \in \ch(\vl)} (\1_{B_t^\wl}(A^\wl_t)\icv_t(\wl) -\icvs_{B_t^\wl}(\wl)) + \sum_{t=1}^T (\theta^\vl(A_t^\vl, B_t^{\ch(\vl)}) - \theta^\vl(A_t^\vl, A_t^{\ch(\vl)})) \eqsp.
\end{align*}
Recall that for any $s<t \in \iint{1}{T}$, we define $\Jmax$ as
\begin{align*}
& J_{\iint{s+1}{s+t}}^\wl = \{l \in \iint{s+1}{s+t} \text{ such that } A_l^\wl \ne B_l^\wl\} \eqsp, \\
    & \Jmax_{\iint{s+1}{s+t}}(\vl) = \{l \in \iint{s+1}{s+t} \text{ such that } l \in J_{\iint{s+1}{s+t}}^\wl \text{ for some } \wl \in \ch(\vl)\} \eqsp,
\end{align*}
and it is showed in \Cref{equation:cardjmaxlittle} during the proof of \Cref{lemma:repeat_assumption} that
\begin{align*}
    \Card(\Jmax_{\iint{s+1}{s+t}}) \leq \sum_{\wl \in \ch(\vl)} \Card(J_{\iint{s+1}{s+t}}^\wl) \leq  t^{\kappa+\eta} \eqsp.
\end{align*}
Under \Cref{assumption:gap_reward}, we can write for the ultimate term in the regret decomposition
\begin{align*}
    |\sum_{t=1}^T (\theta^\vl(A_t^\vl, B_t^{\ch(\vl)}) - \theta^\vl(A_t^\vl, A_t^{\ch(\vl)}))| \leq \Card(\Jmax_{\iint{1}{T}}) \Tilde{\Delta}^\vl_b \leq \Tilde{\Delta}^\vl_b  \cdot T^{\kappa+\eta} \eqsp.
\end{align*}
We also have that--using again the bound on $\Card(\Jmax_{\iint{s+1}{s+t}})$ in the second line
\begin{align*}
    \sum_{t=1}^T \sum_{\wl \in \ch(\vl)} (\icvs_{B_t^\wl}(\wl)-\1_{B_t^\wl}(A^\wl_t)\icv_t(\wl)) & \leq \sum_{t \in \iint{1}{T}, \wl \in\ch(\vl) \colon \icvs_{B_t^\wl}(\wl) >\icv_t(\wl)} (\icvs_{B_t^\wl}(\wl)-\1_{B_t^\wl}(A^\wl_t)\icv_t(\wl)) \\
    & \leq  T^{\kappa+\eta} + \sum_{t \in \iint{1}{T}, \wl \in\ch(\vl) \colon \icvs_{B_t^\wl}(\wl) >\icv_t(\wl), A_t^\wl = B_t^\wl} (\icvs_{B_t^\wl}(\wl)-\icv_t(\wl)) \eqsp,
\end{align*}
but for each step $t \in \iint{1}{T}$ and player $\wl \in \ch(\vl)$ such that $\icvs_{B_t^\wl}(\wl) >\icv_t(\wl)$ and $A_t^\wl = B_t^\wl$, player $\wl$ incurs a regret at least $\icvs_{B_t^\wl}(\wl)-\icv_t(\wl)$, and hence
\begin{align*}
    \sum_{t=1}^T \sum_{\wl \in \ch(\vl)} (\icvs_{B_t^\wl}(\wl)-\1_{B_t^\wl}(A^\wl_t)\icv_t(\wl)) & \leq T^{\kappa+\eta} + \sum_{\wl \in \ch(\vl)} \regret^\wl(T) \eqsp.
\end{align*}
Considering the first term of the regret, we have that
\begin{align*}
    \Jmax_{\iint{1}{T}}(\pa(\vl)) = \{l \in \iint{1}{T} \text{ such that } l \in J_{\iint{1}{T}}^\wl \text{ for some } \wl \in \ch(\pa(\vl))\} \eqsp,
\end{align*}
which includes all the steps of deviation of $\vl$ from action $B_t^\vl$. Therefore, we have that
\begin{align*}
    \Card\{t \in \iint{1}{T}\colon A_t^\vl\ne B_t^\vl\} \leq T^{\kappa +\eta} \eqsp, 
\end{align*}
and at the same time, under the event $\cE$, we have that for $t \geq \wait, \icv_t(\vl) = \icvest_{B_t^\vl}(\wl) \geq \icvs_{B_t^\vl}(\wl)$ (with $\wait$ the time players use for exploration), which makes--by definition of the optimal incentive--$B_t^\vl$ the best action to play in hindsight. Therefore, we have that
\begin{align*}
& \sum_{t=1}^{T} \max_{a} \left\{\mu^{\star,\vl}(a) + \1_{B_t^\vl}(a) \, \icv_t(\vl) -\mu^\pl(A_t^\vl, B_t^{\ch(\vl)}) - \1_{B_t^\vl}(A_t^\vl)\, \icv_t(\vl)\right\} \geq - 2 \, \wait - 2 \cdot T^{\kappa +\eta} \\
& \quad + \sum_{t=1}^T \{\mu^{\star,\vl}(B_t^\vl) + \icvest_{B_t^\vl}(\vl) - \mu^\pl(B_t^\vl, B_t^{\ch(\vl)}) - \icvest_{B_t^vl}(\vl))\} \\
& = -2 \, \wait  -2 \cdot T^{\kappa +\eta} + \sum_{t=1}^T \{\mu^{\star,\vl}(B_t^\vl) - \mu^\pl(B_t^\vl, B_t^{\ch(\vl)})\} \\
& \geq -2 \, \wait - 2 \cdot T^{\kappa +\eta} + \min_{\wl \in \ch(\vl)} \Delta_{b, \wl}^\vl \sum_{t=1}^T \sum_{\wl \in \ch(\vl)} \1\{B_t^\wl \ne b^{\star, \wl}\} \eqsp,
\end{align*}
and we can group all the terms together to obtain
\begin{equation}\label{equation:gfzogh}
\begin{aligned}
    \min_{\wl \in \ch(\vl)} \Delta_{b, \wl}^\vl \sum_{t=1}^T \sum_{\wl \in \ch(\vl)} \1\{B_t^\wl \ne b^{\star, \wl}\} \leq \regret^\vl(T) + \Tilde{\Delta}^\vl_b \cdot T^{\kappa+\eta} + \sum_{\wl \in \ch(\vl)} \regret^\wl(T) + T^{\kappa+\eta} + 2 \wait +2 \cdot T^{\kappa +\eta} \eqsp,
\end{aligned}
\end{equation}
and hence
\begin{align}\label{equation:eghregjzglj}
    \sum_{t=1}^T \sum_{\wl \in \ch(\vl)} \1\{B_t^\wl \ne b^{\star, \wl}\} = o(T) \eqsp,
\end{align}
under the event $\cE$. Now turning to the differences of incentives, we can write
\begin{align*}
\sum_{t=1}^T \sum_{\wl \in\ch(\vl)} \1\{B_t^\wl=b^{\star,\wl}\}\,|\icv_t(\wl)-\icvs(\wl)| & \leq \wait + \sum_{t=\wait+1}^T \sum_{\wl \in\ch(\vl)} |\icvs_{B^\wl_t}(\wl)-\icvest_{B^\wl_t}(\wl)| \\
& \leq \wait + (T-\wait) \cdot \mB(4/T^\beta +\acst \, \mB\, T^{-\extra}) \\
& = o(T) \eqsp,
\end{align*}
under the event $\cE$. Finally, still under $\cE$, we have that the offered incentives make recommended action $B_t^\vl$ optimal. Therefore, as we have
\begin{align*}
        \Card\{t \in \iint{1}{T}\colon A_t^\vl\ne B_t^\vl\} \leq T^{\kappa +\eta} \eqsp,
\end{align*}
and by definition $a^{\star,\vl}$ which boils down to accepting the proposed action--as showed in \Cref{equation:optimal_action}--we can write
\begin{align}\label{equation:egjemghjg}
    \sum_{t=1}^T \1\{A_t^\vl \neq a^{\star,\vl}\} \leq T^{\kappa +\eta} \eqsp.
\end{align}
Combining together \eqref{equation:eghregjzglj}, \eqref{equation:gfzogh} and \eqref{equation:egjemghjg}, we obtain that with probability $1-o(1)$
\begin{align*}
    \frac{1}{T} \sum_{t=1}^T \sum_{\wl \in \ch(\vl)} \1\{B_t^\wl \ne b^{\star, \wl}\} + \frac{1}{T} \sum_{t=1}^T \sum_{\wl \in\ch(\vl)} \1\{B_t^\wl=b^{\star,\wl}\}\,|\icv_t(\wl)-\icvs(\wl)| + \frac{1}{T} \sum_{t=1}^T \1\{A_t^\vl \neq a^{\star,\vl}\} = o(1) \eqsp,
\end{align*}
which allows to conclude that if all players run $\alg$, with probability at least $1-o(1)$
\begin{align*}
    W_1(\delta_{x^{\star,\vl}}, \hat{x}_t^\vl) = o(1) \eqsp.
\end{align*}

\end{proof}

\newpage

\end{document}